\def\IEEEsubmission{0}

\def\complexNumbers{\mathbb{C}}

\def\positiveIntegers{\mathbb{Z}^+}

\def\primaryDataSignal{x}
\def\primaryDataSymbol{S}

\def\primaryOFDMSignal{X}

\def\backscatterOutputSignal{x_\mathrm{out}}

\def\backscatterDataSignal{b}
\def\BDdataFreq{B}

\def\primaryReceivedSignal{y}
\def\DTreceivedPRxSignal{y_\mathrm{p}}
\def\DTreceivedDirectSignal{y_\mathrm{d}}
\def\DTreceivedBDSignal{y_\mathrm{b}}

\def\ReceivedSignalFreq{Y}

\def\DreceivedDirectSignalFreq{\ReceivedSignalFreq_\mathrm{d}}
\def\DreceivedBDSignalFreq{\ReceivedSignalFreq_\mathrm{b}}

\def\timeDomainChannel{h}

\def\directChannel{\timeDomainChannel_\mathrm{d}}
\def\forwardChannel{\timeDomainChannel_\mathrm{f}}
\def\backwardChannel{\timeDomainChannel_\mathrm{b}}
\def\directChannelFreq{H_\mathrm{d}}
\def\forwardChannelFreq{H_\mathrm{f}}
\def\backwardChannelFreq{H_\mathrm{b}}
\def\channelLength{\mathcal{L}}

\def\ChannelComplexAmplitude{\gamma}
\def\timeDelay{\tau}

\def\ChannelPathIndex{i}

\def\noiseSignalTime{w}
\def\noiseSignalFreq{W}
\def\noiseSignalFreqDirect{W_\mathrm{d}}
\def\noiseSignalFreqBD{W_\mathrm{b}}
\def\noisPower{\sigma_w^2}

\def\TotalSubcarriers{N}
\def\ActiveSubcarriers{K}

\def\CPlength{N_\mathrm{CP}}

\def\dataSubcarrierSet{\mathcal{K}_\mathrm{d}}
\def\nullSubcarrierSet{\mathcal{K}_\mathrm{b}}
\def\fskOneset{\mathcal{K}_\mathrm{b_0}}
\def\fskTwoset{\mathcal{K}_\mathrm{b_1}}

\def\TotalDataSubcarriers{\TotalSubcarriers_\mathrm{d}}
\def\TotalBDSubcarriers{\TotalSubcarriers_\mathrm{b}}

\def\LoadImpedance{Z_\mathrm{L}}
\def\AntennaImpedance{Z_\mathrm{A}}

\def\AntennaImpedancePhase{\theta_\mathrm{A}}
\def\LoadImpedancePhase{\theta_\mathrm{L}}
\def\phase#1{\theta_\mathrm{#1}}
\def\reflectionCoefficient{\Gamma_{\rm b}}

\def\BDFreqShift{f_\zeta}

\def\pfa{P_\mathrm{FA}}
\def\pmd{P_\mathrm{MD}}
\def\poe {P_\mathrm{e}}
\def\probability#1{\mathrm{Pr}(#1)}
\def\threshold{\eta}
\def\cdf#1#2{F_{#1}(#2)}
\def\pdf#1#2{f_{#1}(#2)}
\def\characfun#1#2{\Phi_{{#1}}(#2)}
\def\ts#1{r_{#1}} 

\if\IEEEsubmission1
\documentclass[journal,12pt,onecolumn,draftclsnofoot]{IEEEtran}
\else
\documentclass[journal]{IEEEtran}
\fi
\usepackage{multicol}
\usepackage{mathtools}
\usepackage{amsthm}
\usepackage{acronym}
\usepackage[bookmarksopen=true]{hyperref}
\usepackage{stfloats}
\usepackage{amsfonts}
\usepackage{acronym}
\usepackage[noadjust]{cite}
\usepackage{multirow}
\usepackage{bm}
\usepackage{enumitem}
\usepackage{amsmath,amssymb}
\usepackage{graphicx}
\usepackage{epstopdf}
\usepackage[table,xcdraw]{xcolor}
\usepackage[caption=false,font=footnotesize]{subfig}
\usepackage[geometry]{ifsym}
\usepackage{array}
\usepackage[utf8]{inputenc}
\usepackage[T1]{fontenc}
\usepackage{pifont}
\usepackage{tabularray}
\usepackage{lipsum}
\def\BibTeX{{\rm B\kern-.05em{\sc i\kern-.025em b}\kern-.08em
		T\kern-.1667em\lower.7ex\hbox{E}\kern-.125emX}}

\newcommand\mydots{\hbox to 1em{.\hss.\hss.}}

\newtheorem{theorem}{Theorem}

\newtheorem{lemma}{Lemma}

\makeatletter
\newif\ifAC@uppercase@first%
\def\Aclp#1{\AC@uppercase@firsttrue\aclp{#1}\AC@uppercase@firstfalse}%
\def\AC@aclp#1{%
	\ifcsname fn@#1@PL\endcsname%
	\ifAC@uppercase@first%
	\expandafter\expandafter\expandafter\MakeUppercase\csname fn@#1@PL\endcsname%
	\else%
	\csname fn@#1@PL\endcsname%
	\fi%
	\else%
	\AC@acl{#1}s%
	\fi%
}%
\def\Acp#1{\AC@uppercase@firsttrue\acp{#1}\AC@uppercase@firstfalse}%
\def\AC@acp#1{%
	\ifcsname fn@#1@PL\endcsname%
	\ifAC@uppercase@first%
	\expandafter\expandafter\expandafter\MakeUppercase\csname fn@#1@PL\endcsname%
	\else%
	\csname fn@#1@PL\endcsname%
	\fi%
	\else%
	\AC@ac{#1}s%
	\fi%
}%
\def\Acfp#1{\AC@uppercase@firsttrue\acfp{#1}\AC@uppercase@firstfalse}%
\def\AC@acfp#1{%
	\ifcsname fn@#1@PL\endcsname%
	\ifAC@uppercase@first%
	\expandafter\expandafter\expandafter\MakeUppercase\csname fn@#1@PL\endcsname%
	\else%
	\csname fn@#1@PL\endcsname%
	\fi%
	\else%
	\AC@acf{#1}s%
	\fi%
}%
\def\Acsp#1{\AC@uppercase@firsttrue\acsp{#1}\AC@uppercase@firstfalse}%
\def\AC@acsp#1{%
	\ifcsname fn@#1@PL\endcsname%
	\ifAC@uppercase@first%
	\expandafter\expandafter\expandafter\MakeUppercase\csname fn@#1@PL\endcsname%
	\else%
	\csname fn@#1@PL\endcsname%
	\fi%
	\else%
	\AC@acs{#1}s%
	\fi%
}%
\edef\AC@uppercase@write{\string\ifAC@uppercase@first\string\expandafter\string\MakeUppercase\string\fi\space}%
\def\AC@acrodef#1[#2]#3{%
	\@bsphack%
	\protected@write\@auxout{}{%
		\string\newacro{#1}[#2]{\AC@uppercase@write #3}%
	}\@esphack%
}%
\def\Acl#1{\AC@uppercase@firsttrue\acl{#1}\AC@uppercase@firstfalse}
\def\Acf#1{\AC@uppercase@firsttrue\acf{#1}\AC@uppercase@firstfalse}
\def\Ac#1{\AC@uppercase@firsttrue\ac{#1}\AC@uppercase@firstfalse}
\def\Acs#1{\AC@uppercase@firsttrue\acs{#1}\AC@uppercase@firstfalse}

\acrodef{BC}{backscatter communications}
\acrodef{SR}{symbiotic radio}
\acrodef{IoT}{Internet-of-Things}
\acrodef{OFDM}{orthogonal frequency division multiplexing}
\acrodef{BS}{base station}
\acrodef{Rx}[RX]{receiver}
\acrodef{BD}{backscatter device}
\acrodef{ICI}{inter-carrrier interference}
\acrodef{5G}{fifth-Generation}
\acrodef{LTE}{long-term evolution}
\acrodef{CP}{cyclic prefix}
\acrodef{FSK}{frequency-shift keying}
\acrodef{PFSK}{pseudo frequency-shift}
\acrodef{BFSK}{binary frequency-shift keying}
\acrodef{FFT}{fast Fourier transform}
\acrodef{DFT}{discrete Fourier transform}
\acrodef{IDFT}{inverse discrete Fourier transform}
\acrodef{IFFT}{inverse-\ac{fft}}
\acrodef{ISI}{inter-symbol interference}
\acrodef{ML}{maximum likelihood}
\acrodef{ZF}{zero-forcing}
\acrodef{MMSE}{minimum mean square error}
\acrodef{SNR}{signal-to-noise-ratio}
\acrodef{BER}{bit-error-rate}
\acrodef{CSCG}{circularly symmetric complex Gaussian}
\acrodef{PDF}{probability density function}
\acrodef{PoE}{probability of error}
\acrodef{RF}{radio frequency}
\acrodef{PMD}{probability of miss detection}
\acrodef{PFA}{probability of false alarm}
\acrodef{ROC}{receiver operating characteristic}
\acrodef{OOK}{on-off keying}
\acrodef{ASK}{amplitude-shift keying}
\acrodef{PD}{probability of detection}
\acrodef{CDF}{cumulative distribution function}
\acrodef{3GPP}{3rd generation partnership project}
\acrodef{BPSK}{binary phase shift keying}
\acrodef{PSK}{phase shift keying}
\acrodef{MSK}{Minimum shift keying}
\acrodef{MCU}{microcontroller unit}
\acrodef{RFID}{radio frequency identification}
\acrodef{CRC}{cyclic redundancy check}
\acrodef{ARQ}{automatic repeat requests}
\acrodef{CFO}{carrier frequency offset}

\begin{document}

\title{
	{Interference-Free Backscatter Communications for OFDM-Based Symbiotic Radio}\\

	\thanks{M. B. Janjua is with the Department of Electrical and Electronics Engineering, Istanbul Medipol University, Istanbul, 34810, Türkiye, and also with the R\&D Department of Oredata, Istanbul, Türkiye (email: bilal.janjua@oredata.com). A. \c{S}ahin is with the Electrical Engineering Department, University of
South Carolina, Columbia, SC, USA. (email: asahin@mailbox.sc.edu). H. Arslan is with the Department of Electrical and Electronics Engineering, Istanbul Medipol University, Istanbul, 34810, Türkiye (email: huseyinarslan@medipol.edu.tr)}
\thanks{This paper has been accepted for publication in part to IEEE Global Communication Conference (GLOBECOM) 2024 \cite{janjua2024}} 

	\author{Muhammad Bilal Janjua,~\IEEEmembership{Student Member,~IEEE,} Alphan~\c{S}ahin,~\IEEEmembership{Member,~IEEE} and H\"{u}seyin Arslan,~\IEEEmembership{Fellow, IEEE}} 
}
\maketitle

\begin{abstract}

This study proposes an orthogonal frequency division multiplexing (OFDM) based scheme to achieve interference-free backscatter communications (BC) in a symbiotic radio system. In this scheme, the backscatter device performs frequency shift keying (FSK) modulation to shift the primary signal, i.e., the OFDM symbols transmitted from a base station, in the frequency domain to transmit its information. Symbiotically, the \acl{BS} empties specific subcarriers within the band so that the received frequency-shifted signals from the backscatter device and the primary signal are always orthogonal. Based on the empty subcarrier placement and the corresponding FSK modulation, we propose three schemes including on-off keying (OOK), FSK-1, and FSK-2. Specifically, the first scheme integrates FSK into OOK modulation while the second and the third schemes are based on the conventional FSK modulation with different in-band null-subcarrier allocation. To address the channel estimation challenge for the signals arriving from a backscatter device, we consider a non-coherent detector for obtaining the information from the backscatter signal at the receiver. We derive the bit-error rate performance of the detector theoretically. The comprehensive simulations show that the proposed approach achieves a lower bit-error rate up to $10^{-4}$ at $30$ dB with BC by eliminating direct link interference. 
\end{abstract}

\section{Introduction}
\acresetall

Next-generation wireless networks are expected to be energy-efficient and sustainable by supporting use cases that 
involve \ac{IoT} devices with very low power consumption \cite{huang2019survey}. In this direction, passive and ambient \ac{IoT} devices were studied in the \Ac{5G} New Radio Release~18 and Release~19~\cite{lin2023bridge}. The ambient devices can be battery-less or have a battery and operate with the harvested energy from radio waves or other sources~\cite{butt2023ambient}. Under such limitations, modulation techniques relying on active components may not be feasible for communication due to their high power consumption, leading to the \ac{BC} paradigm that is widely used in \ac{RFID} systems for low-power information transfer~\cite{epcglobalradio,han2017wirelessly,niu2019overview,  liu2019next,stanacevic2020backscatter}. With \ac{BC}, the data is transmitted by modulating the response of the antenna to the incoming  \ac{RF} signals. {\color{black}In the \ac{RFID} systems, a dedicated \ac{RF} signal exciter is deployed to transmit unmodulated carrier signals for \ac{BC}. As the coverage area and the number of tags (i.e., \ac{BD}) increases, more carrier exciters are needed to serve all the \acp{BD} in the environment. Considering the limitations of \ac{RFID} systems, an ambient \ac{BC} is developed in which a \ac{BD} utilizes the signals transmitted by existing wireless systems as a carrier for data transmission \cite{liu2013ambient}. The ambient \ac{BC} eliminates the need for a dedicated device for carrier signal transmission and opens a new era of \ac{BC}. However, the reliability of ambient \ac{BC} is very low due to the interference between the unknown strong \ac{RF} and weak backscatter signal at the intended \ac{Rx}.} In this study, we address this interference challenge with the concept of \ac{SR} by designing the incoming \ac{RF} signals.

The idea of \ac{SR} for \ac{BC} in passive \ac{IoT} was first thoroughly investigated in \cite{long2019symbiotic}, where the \ac{BD} was considered a parasitic entity in a primary system, e.g., an existing wireless system. In contrast to ambient \ac{BC}, \ac{BD} is regarded as a component of the primary system in \ac{SR} and shares radio resources with the primary user, with the primary transmitter being designed to support its requirements. In \cite{liang2020symbiotic}, the authors show that \ac{SR} can exploit the benefits of ambient \ac{BC} and cognitive radio since \ac{SR} can share share the spectrum between primary users and \ac{BD}. Also, it can provide highly reliable \ac{BC} through a joint decoder at the receiver. The authors in \cite{zhang2021mutualistic} and \cite{xu2021enabling} analyze the mutualistic relationship between \ac{BD} and a primary system in \ac{SR}, where both systems benefit each other by sharing the radio resource. Some other \ac{SR}-based relationships between low-power \ac{IoT} devices and existing wireless systems are explored in \cite{zhang20196g, bariah2020prospective,chen2020vision,nawaz2021non,janjua2023survey}. These earlier studies suggest that by cooperative resource sharing and system design, \ac{SR} can avoid the need for new infrastructure to support \ac{BC} and may allow the network to overcome the coexistence issues. Also, \ac{SR} can mitigate interference between the primary and backscatter signals. Since most wireless systems, including cellular and Wi-Fi networks, utilize \ac{OFDM} as their standard waveform, an opportunity arises to design the \ac{OFDM} signal in \ac{SR} to address the interference issue in \ac{BC} concerning ambient \ac{IoT} devices.

Exploiting the \ac{OFDM} signal for \ac{BC} is critical for low-power passive \ac{IoT} devices as they coexist with the existing wireless systems. Within this context, multiple studies consider \ac{OFDM}-based \ac{BC} \cite{yang2016backscatter, elmossallamy2019noncoherenta, elmossallamy2018backscatter, nemati2020subcarrier, takahashi2019ambient, hara2021ambient, chen2023pilot}. These studies explore the cyclostationary features of \ac{OFDM} signal generated by an ambient wireless system for \ac{BC}. In~\cite{yang2016backscatter}, a backscatter modulation scheme is designed for \ac{BC} over \ac{OFDM} signal in the time domain. The uncorrupted part of the \ac{CP} is processed at the receiver for direct link interference cancellation and \ac{BD} signal detection. Still, the receiver must know the channel's lengths for this scheme to work in practical scenarios. Also, the waveform designed for \ac{BC} causes variation in the primary signal, which may lead to interference at the primary receiver. The guard bands in the \ac{OFDM} signal are used for \ac{BC} in \cite{elmossallamy2018backscatter, elmossallamy2019noncoherenta, elmossallamy2019noncoherentb}. Specifically, a \ac{BD} applies \ac{FSK} to shift the spectrum of the backscatter signal to the out-of-band region of the primary \ac{OFDM} signal to transmit its information. This approach cannot be directly implemented in \ac{SR} as the joint receiver needs to scan the adjacent channel to detect the data of \ac{BD}, which requires an additional filter to capture the backscatter signal in the guard subcarriers. Consequently, this solution adds complexity to the receiver. The authors in \cite{nemati2020subcarrier} introduce the subcarrier-wise ambient \ac{BC} and propose to transmit one-bit information at each subcarrier. This scheme requires a bank of passive notch filters at the \ac{BD},  significantly increasing its hardware complexity. The preamble and pilots in \ac{OFDM} signals are used for ambient \ac{BC} in \cite{takahashi2019ambient,hara2021ambient, chen2023pilot, liao2023band}, but utilizing the preamble and pilots degrades the channel estimation performance at the primary receiver, which reduces the reliability of the whole system. Additionally, this approach requires strict synchronization at the backscatter receiver. Besides, in \cite{ liao2023band}, the authors propose to use the cell-specific reference signals in the \ac{LTE} system to enable \ac{BC}. In particular, the \ac{BD} applies \ac{FSK} to create an artificial Doppler in the incident \ac{LTE} signal and shift the reference signals to another subcarrier to transmit its information. To this end, the \ac{BD} and 
backscatter receiver must have prior information on  \ac{LTE} reference signals. Nevertheless, these studies do not mainly design the \ac{OFDM} symbols to achieve reliable \ac{BC} for an \ac{SR} system. For instance, it is possible to arrange the data, pilot, and null subcarriers at the primary system to receive the direct path and \ac{BC} signals over different subcarriers if \ac{BD} is capable of manipulating the signal frequency \cite{bharadia2015backfi}. This arrangement can be achieved without changing the transceiver design of the primary system. Furthermore, such a system can support passive \ac{IoT} devices and ambient power-enabled \ac{IoT} devices in various wireless networks, including Wi-Fi, \ac{LTE}, \ac{5G}, and beyond.

\subsection{Contributions}

 In this work, we propose an \ac{OFDM}-based scheme in the context of \ac{SR}. Our contributions can be listed as follows:
 \begin{itemize}

 \item We design the \ac{OFDM} signal transmitted from a \ac{BS} such that it is always orthogonal to the backscatter signal within the band. To this end, we utilize \ac{FSK} modulations at the \ac{BD} and shift the data symbols to the subcarriers dedicated to \ac{BD} in the frequency domain. As a result, the primary user's and BD's data are received on separate subcarriers within the received \ac{OFDM} signal, and the interference originating from the direct link is circumvented. 
 \item By extending our preliminary work in \cite{janjua2024}, we analyze and compare the \ac{OOK} modulation with \ac{FSK} modulation schemes for \ac{BC}. The proposed \ac{OOK} waveform shifts the data symbols to the subcarriers dedicated to \ac{BD} during on period. We provide the theoretical derivation for the error rate based on Lemma \ref{lemma:1}. We evaluate the retransmission probability of the proposed schemes based on the \ac{RFID} Gen2 standard $5$-bit \ac{CRC} encoder/decoder. 
 \item We show the efficacy of the non-coherent detector in receiving \ac{BD}’s data at \ac{Rx}. We analyze the error-rate performance of the non-coherent detector analytically and assess the scheme via comprehensive simulations. We also demonstrate the impact of hardware impairments, specifically \ac{CFO}, on the efficiency of \ac{FSK}-based \ac{BC} schemes.     
\end{itemize}

{\em Organization:} The rest of the paper is organized as follows.  Section~\ref{sec:system_design} elaborates on the \ac{SR} system design and the notations and preliminaries used in the rest of the sections. Section~\ref{sec:scheme} discusses the proposed schemes in detail. In Section~\ref{sec:detection}, non-coherent detection and performance analysis are presented. Section~\ref{sec:results} provides the simulation and theoretical results and detailed discussions. The concluding remarks are given in Section~\ref{sec:conclusion}.
 
 {\em Notation:} The sets of complex numbers and positive integers are denoted by $\complexNumbers$ and $\positiveIntegers$, respectively. $E[\cdot]$ represents the expectation of its argument over random variables. $\mathcal{CN}(0,\sigma^2)$ is the circularly symmetric complex Gaussian distribution with zero mean and $\sigma^2$ variance. $|a|$ denotes the cardinality of a set $a$. $\mathrm{Pr}(A|B)$ represents the probability of an event $A$ given the event $B$. $\mathrm{Pr}(A)$ represents the probability of an event $A$. $\lfloor{\cdot}\rfloor$ represents the floor function. $\circledast$ is the circular convolution operation, and $*$ represents the linear convolution.

\section{Symbiotic Radio System Model}
\label{sec:system_design}

Consider an \ac{OFDM}-based \ac{SR} system consisting of two primary devices, i.e., a \ac{BS}, an \ac{Rx},  and a \ac{BD}, as shown in \figurename~\ref{fig:system_model}, where \ac{BS} and \ac{Rx} are assumed to be single-antenna active devices, and \ac{BD} is a single-antenna passive device. The \ac{BD} is a low-power, low-cost, and low-complexity device that belongs to a passive \ac{IoT} system. The primary signal transmitted by \ac{BS} is utilized as \ac{BD}'s carrier signal to transmit its information using \ac{BC} to the \ac{Rx}. Furthermore, the primary signal is a modulated signal containing the information for the \ac{Rx}, and the \ac{Rx}'s goal is to obtain the information that is transmitted from both \ac{BD} and \ac{BS}.

\begin{figure}
    \centering
    \includegraphics[width = \linewidth]{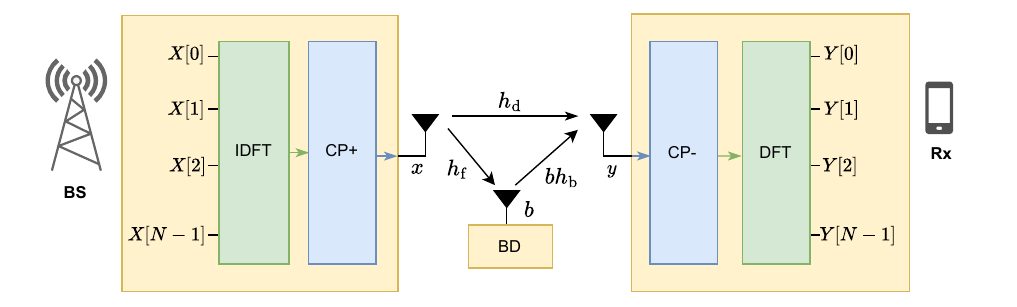}
    \caption{The system model with single BS, BD, and a joint RX.}
    \label{fig:system_model}
\end{figure}

\subsection{Primary Signal}
 The primary signal is an \ac{OFDM} symbol with $\ActiveSubcarriers$ active subcarriers containing data symbols for \ac{Rx}. Let $\primaryOFDMSignal[k]$ be the $k$-th data symbol at \ac{BS}. After an \ac{IDFT} is applied to the data symbols, the $n$-th time domain sample of the corresponding \ac{OFDM} symbol, denoted by $x[n]$, can be expressed as
\begin{equation}
    x[n]=\frac{1}{\TotalSubcarriers}\sum_{k=0}^{\ActiveSubcarriers-1}X[k]e^{\frac{j2\pi kn}{\TotalSubcarriers}}~,~n\in \{0,1,\cdots, \TotalSubcarriers-1\}~,
\end{equation}
where $\TotalSubcarriers$ is the \ac{IDFT} size. Afterward, a \ac{CP} of length $\CPlength$ that is larger than the maximum delay spread is added to the time domain signal, and the corresponding \ac{OFDM} samples are transmitted over a wireless channel. The bandpass signal transmitted by the \ac{BS} can be written as 
\begin{equation}
    x(t) =  x_\mathrm{bb}(t)e^{j2\pi f_\mathrm{c}t}~,
\end{equation}
where $x_\mathrm{bb}(t)$ is the baseband \ac{OFDM} signal with \ac{CP} in the time domain, and $f_\mathrm{c}$ denotes the carrier frequency.

\subsection{BD Signal}
\label{subsec: BD Signal}

Conventionally, the load modulation is applied by switching the \ac{BD}'s antenna between non-reflecting and reflecting states for transmitting $'0'$ and $'1'$, respectively~\cite{epcglobalradio}. The achievement of the reflecting state involves impedance matching, which is realized by connecting an antenna with an impedance of $\AntennaImpedance = |\AntennaImpedance|e^{j\AntennaImpedancePhase}$ to a load with an impedance of $\LoadImpedance = |\LoadImpedance|e^{j\LoadImpedancePhase}$, ensuring $\LoadImpedance = \AntennaImpedance$. As the $\LoadImpedance$ matches with the $\AntennaImpedance$, maximum power transfers from the source (i.e., antenna) to the load, which is used to energy harvesting or detect the information in the signal. Conversely, impedance mismatching is implemented to obtain a non-reflecting state by connecting the antenna to the load in such a way that $\LoadImpedance \neq \AntennaImpedance$. Based on the mismatch between the $\LoadImpedance$ and $\AntennaImpedance$, power is reflected from the source, which is used to backscatter the incident signal. It is noted that the amount of power transfer between the source and the load is related to the impedance matching and mismatch such that the perfect matching leads to maximum power transfer from the load to the source, whereas the perfect mismatch results in minimal power transfer to the load and maximum power reflection from the source. Hence, the reflection coefficient $\reflectionCoefficient$ for a backscatter modulator is expressed as \cite{xu2018practical}:
\begin{equation}
    \reflectionCoefficient = \frac{\LoadImpedance-\AntennaImpedance^*}{\LoadImpedance-\AntennaImpedance}=|\reflectionCoefficient| e^{j\phase{b}}~.
\end{equation}
Different values of $\reflectionCoefficient$ can be obtained by changing the values of $\LoadImpedance$ because the value of $\AntennaImpedance$ is constant based on the antenna structure. The amplitude and phase of the backscatter signal are described by $|\reflectionCoefficient|$ and phase $\phase{b}$, which can be defined as \cite{wu2022survey}: 
\begin{equation}
    |\reflectionCoefficient| = \frac{|\AntennaImpedance|+|\LoadImpedance|-2|\AntennaImpedance||\LoadImpedance|\mathrm{cos}(\AntennaImpedancePhase-\LoadImpedancePhase)}{|\AntennaImpedance|+|\LoadImpedance|+2||\AntennaImpedance||\LoadImpedance|\mathrm{cos}(\AntennaImpedancePhase-\LoadImpedancePhase)}~,
\end{equation}
\begin{equation}
    \phase{b} = \mathrm{arctan}\left(\frac{2|\AntennaImpedance||\LoadImpedance|\mathrm{sin}(\AntennaImpedancePhase-\LoadImpedancePhase)}{|\AntennaImpedance|^2|\LoadImpedance|^2}\right)~.
\end{equation}
The impedance of the \ac{BD} can be manipulated to attain complex values by carefully selecting appropriate inductance values. By connecting an RF switch to the loads $\LoadImpedance$ and alternately switching between them, the BD can effectively replicate the complex signal $j2\pi\BDFreqShift t$ with a frequency of $\BDFreqShift$ \cite{iyer2016inter,vougioukas2018switching}. Here, $\zeta$ represents the frequency index for shifting the primary signal's data symbols. 
By precisely controlling the amplitude and frequency of the reflected signal, both \ac{OOK} and \ac{FSK} modulations can be generated to transmit the information from the \ac{BD}.

\par Conventionally, a \ac{BD} modulates the incoming signal with a rectangular pulse generated by switching between two loads~\cite{epcglobalradio}. Such a modulation technique is not preferable for a receiver in \ac{SR} as changing load sharply requires \ac{BD} to transmit in the out-of-band region of the primary communication system. One of the methods to tackle this concern is pulse-shaping, which involves continuous load variations at \ac{BD} for \ac{BC} \cite{kimionis2016pulse}. A \ac{BD} consists of a variable load $\LoadImpedance$ (i.e., diode) controlled by a \ac{MCU}. The \ac{MCU} applies different voltage levels to the diode to perform continuous load modulations. This approach suppresses the out-of-band emissions and allows the synthesis of specific signals like a complex exponential for frequency modulation or shapes the time-frequency characteristics of the reflected signal. Besides, other \ac{BD} hardware architectures and designs are available in the literature that can generate complex signals to achieve single side-band \ac{FSK} without higher order and mirror harmonics \cite{ iyer2016inter}. As \ac{BD} design is not part of this study, interested readers are directed to the works in \cite{ding2020harmonic} and \cite{alhassoun2023spectrally}.

\par In this study, we consider a \ac{BD} consisting of multiple loads $\LoadImpedance$ (i.e., diode or transistors) and an \ac{MCU} discussed in \cite{iyer2016inter, wang2020low,  ding2020harmonic},  which can perform single side-band  \ac{FSK} and \ac{OOK} modulations. We assume that the \ac{BD}'s \ac{MCU} and other low-power operations are supported by harvesting ambient sources like \ac{RF} signals or solar energy \cite{long2019symbiotic}. However, the availability of harvested energy may vary in realistic environments depending on the available sources. A separate research study could be conducted to analyze the energy harvesting efficiency in such scenarios and develop sustainable energy harvesting methods for \ac{BD}. Let \ac{BD} use frequency modulation to transmit information. Then, the frequency-modulated primary signal at the output of \ac{BD} can be written as
\begin{equation}
    \backscatterOutputSignal (t) = (\forwardChannel(t) \ast \primaryDataSignal(t))\reflectionCoefficient\backscatterDataSignal(t)~,
\end{equation}
where $\backscatterDataSignal(t)$ is the signal generated at the \ac{BD}, and 
$\forwardChannel(t)$ represents the channel on the forward link (i.e., the link between \ac{BS} and \ac{BD}). 
Note that the implementations of single side-band \ac{FSK} discussed in \cite{iyer2016inter} and \cite{wang2020low} can be used to obtain  $\reflectionCoefficient\backscatterDataSignal(t)$ in practice. The rationale for employing \ac{FSK} modulation is to enable the orthogonal reception of the primary signal and the \ac{BD} signal at the receiver. Alternative modulation schemes, such as \ac{PSK} and \ac{MSK}, could potentially be used for the \ac{BD} signal. However, these other modulation approaches would come at the cost of performance degradation, loss of orthogonality between the primary and \ac{BD} signals, and increased receiver complexity.

\subsection{Received Signal}

 The signal transmitted by the \ac{BS} is received at the \ac{Rx} from the direct link and the backscatter link. The superposed signal in the baseband can be expressed as
\begin{equation}
    \primaryReceivedSignal(t) = \directChannel(t) \ast \primaryDataSignal(t)+\backwardChannel(t) \ast \backscatterOutputSignal (t)+\noiseSignalTime(t)~,
\end{equation}
where $\primaryDataSignal(t)$ is the \ac{Rx} data signal,  $\backscatterDataSignal(t)$ is the \ac{BD} data signal, $\ast$ denotes the convolution operation, and $\noiseSignalTime(t)$ is the additive white Gaussian noise. The functions $\directChannel(t)$ and $\backwardChannel(t)$ represent the continuous-time channel for the direct link (between \ac{BS} and \ac{Rx}) and backward link (between \ac{BD} and \ac{Rx}), respectively. Each channel path follows the multi-path Rayleigh fading, i.e., $\timeDomainChannel(\timeDelay)=\sum_{\ChannelPathIndex=0}^{\channelLength-1}\ChannelComplexAmplitude_\ChannelPathIndex\delta(\timeDelay-\timeDelay_\ChannelPathIndex)$ with complex amplitude $\ChannelComplexAmplitude_\ChannelPathIndex$ and delay $\timeDelay_\ChannelPathIndex$ of the $i$th path. Note that $\ChannelComplexAmplitude_\ChannelPathIndex$ also includes the shadow fading occurs in the wireless channel. The maximum delay excess delay of the composite channel is $\timeDelay_\mathrm{max} = \max\{\timeDelay_\mathrm{d}, \timeDelay_\mathrm{f}+\timeDelay_\mathrm{b}\}$, where $\timeDelay_\mathrm{d}$, $\timeDelay_\mathrm{f}$, and  $\timeDelay_\mathrm{b}$ represent the time delay corresponding to direct, forward, and backward links. After the \ac{CP} is discarded, the received signal at the \ac{Rx} in the discrete-time can be written as 
\begin{equation}
    \primaryReceivedSignal [n] = \DTreceivedDirectSignal[n]+\DTreceivedBDSignal[n]+\noiseSignalTime[n]~,
    \label{Eq:prx_received}
\end{equation}
where $\DTreceivedDirectSignal[n]=\directChannel[n]\ast \primaryDataSignal[n]$ is the signal received from the \ac{BS} through the direct link, $\DTreceivedBDSignal[n] = \forwardChannel[n] \ast \left(\backwardChannel[n]\ast(\primaryDataSignal[n] \reflectionCoefficient\backscatterDataSignal[n]\right)$ is the signal received from the \ac{BD} over the backscatter link, $\directChannel[n]$, $\forwardChannel[n]$, $\backwardChannel[n]$, $\primaryDataSignal[n]$, and $\backscatterDataSignal[n]$ represent the direct channel, forward channel, backward channel, primary signal, and backscatter signal in the discrete-time domain, respectively, and $\noiseSignalTime[n]\sim \mathcal{CN}(0,\noisPower)$ is the additive white Gaussian noise with zero mean and $\noisPower$ variance.

We note from \eqref{Eq:prx_received} that if \ac{OOK} modulation is used at \ac{BD} and the \ac{Rx} receives the superposed $\DTreceivedDirectSignal$ and $\DTreceivedBDSignal$ signals over the same frequency band, i.e., the signal $\DTreceivedDirectSignal$ interferes with the signal $\DTreceivedBDSignal$ due to high channel gains. Then, the \ac{Rx} must decode the primary signal and perform successive interference cancellation before detecting the \ac{BD} information from $\DTreceivedPRxSignal$. The primary signal decoding is not feasible in ambient \ac{BC} systems as the primary data is unknown to the \ac{Rx}. In an \ac{SR} system, such a decoding process degrades the performance of the \ac{BD} detector significantly and negatively impacts the reliability of the \ac{BC}. Moreover, successive interference cancellation increases the complexity of the \ac{Rx}. On the other hand, in \ac{FSK}-based approaches used in ambient \ac{BC} \cite{elmossallamy2019noncoherenta}, an unutilized frequency band other than the band of the primary signal is required to shift the primary signal and avoid direct link interference. To achieve a reliable \ac{BC} in an \ac{OFDM}-based \ac{SR} system, we propose a scheme to prevent direct link interference at the \ac{Rx}, which will be discussed in the next section. 

\section{Proposed Schemes for Interference-Free BC}
\label{sec:scheme}

The conventional \ac{5G} and Wi-Fi systems use \ac{OFDM} as the standard waveform \cite{farhang2016ofdm, ozdemir2007channel}, in which data symbols are placed in the frequency domain over the narrow band subcarriers. In the \ac{OFDM}-based \ac{SR}, \ac{BD} modulates the incoming signal by changing its amplitude, phase, or frequency shift. In the case of amplitude and phase modulation \ac{BD}, \ac{Rx} receives the primary signal with modified amplitude or phase, but data symbols in the primary signal coming from backscatter and direct path interfere. However, one way to avoid this interference is to shift the data in the primary signal in the frequency with \ac{FSK} at the \ac{BD}. For instance, if subcarriers exist in the \ac{OFDM} symbol, \ac{BD} can shift the primary data to those subcarriers to prevent the backscatter signal from the direct path signal interference and vice versa. By exploiting this approach, we propose the placement of null subcarriers within the \ac{OFDM} symbol for interference-free \ac{BC}. Let $\TotalBDSubcarriers$ and $\TotalDataSubcarriers$ denote the number of subcarriers for \ac{BD} and \ac{Rx}, respectively. To receive the \ac{BD} information with \ac{OOK}, $\TotalBDSubcarriers$ must be equal to $\TotalDataSubcarriers$; however, $\TotalBDSubcarriers$ must be greater than $\TotalDataSubcarriers$ to receive the \ac{BD} information with \ac{FSK}.

\subsection{Proposed Scheme for OOK Modulation}
\Ac{OOK} is an \ac{ASK} modulation in which the information is represented with two amplitudes, i.e., ON and OFF. In terms of \ac{BC}, \ac{OOK} is performed by switching the antenna between reflecting and non-reflecting states to cause a change in the backscatter signal \cite{devineni2018ambient}. In the proposed system, \ac{BD} also switches reflecting states at a different rate to cause a frequency shift in the backscatter signal while transmitting information bit $'1'$. The \ac{BD} waveform can be expressed as 
\begin{equation}
    \backscatterDataSignal (t) = 
    \begin{cases}
        0, & b=0\\
     e^{j2\pi \BDFreqShift t}, & b=1
    \end{cases}~.
    \label{Eq:bd_ook}
\end{equation}
where $\BDFreqShift=\zeta\Delta f$ with $\Delta f$ representing the subcarrier spacing in the \ac{OFDM} signal. Unlike the conventional \ac{OOK} modulation, the \ac{BD} shifts the data in the primary signal to transmit bit $'1'$ and remain constant for bit $'0'$. The \ac{Rx} detects the \ac{BD}'s information from the null subcarriers.

In the proposed approach for \ac{OOK}-based \ac{BC} without interference, null subcarriers are placed between adjacent data subcarriers, as illustrated in \figurename~\ref{fig:proposed-schemes}. The subcarrier allocation for primary data and \ac{BC} at the \ac{BS} can be expressed as  
\begin{equation}
    X[k] = 
    \begin{cases}
        \textbf{\primaryDataSymbol}[m], & k = 2m\\
        0, & \mathrm{otherwise}
    \end{cases}~,
\end{equation}
where $m\in\mathbb{Z}^+$ with $m\leq\frac{\TotalSubcarriers}{2}-1$, and  $\textbf{\primaryDataSymbol}[m]=[S[0],S[1],\cdots,S[\TotalDataSubcarriers-1]]$ with $\TotalDataSubcarriers=\frac{\TotalSubcarriers}{2}$ denotes the primary data symbol. In the \ac{OOK} scheme, the \ac{BD} shifts the data in the incoming signal with frequency $\BDFreqShift$ with $\zeta = 1$ to transmit bit $'1'$ over the null subcarrier. The \ac{Rx} receives the \ac{BD} bits and primary data symbols over null subcarriers and data subcarriers, respectively, as illustrated in  \figurename~\ref{fig:proposed-schemes}. Thus, the \ac{OOK} scheme prevents interference from the primary signal to the \ac{BD} signal at \ac{Rx}. Here, we mentioned only single data and corresponding null subcarrier placement. However, different arrangements can be made to allocate more than one data subcarrier together. For instance, two data subcarriers followed by two null subcarriers can also be placed to achieve the \ac{BC} without interference at the \ac{Rx}. In that case, a general subcarrier allocation at the \ac{BS} for \ac{OOK} modulation can be expressed as  
\begin{equation}
    X[k] = 
    \begin{cases}
        \textbf{\primaryDataSymbol}[m/\zeta], & k \in m\\
        0, & \mathrm{otherwise}
    \end{cases}~,
    \label{Eq:OOK-2}
\end{equation}
where $\zeta\in [2,4,\cdots,\frac{\TotalSubcarriers}{2}]$ and subcarrier index $m=\{0,\zeta, \zeta+2, \cdots, \frac{\TotalSubcarriers}{2}-\zeta\}$. Now, to prevent direct link interference when two data subcarriers are placed next to each other, the \ac{BD} shifts the signal with frequency $\zeta=2$ to transmit bit $'1'$. 

\begin{figure}
    \centering
    \includegraphics[width=\linewidth]{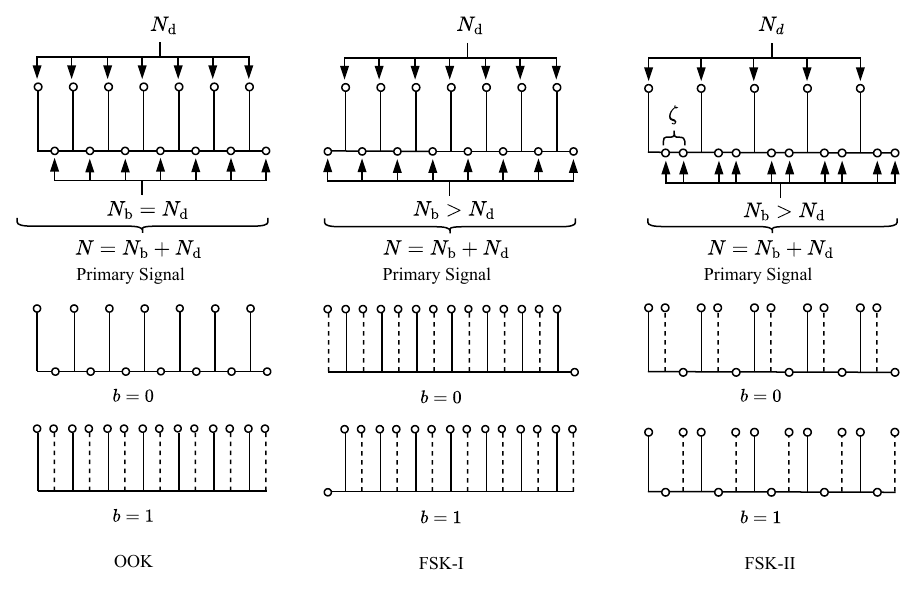}
    \caption{The proposed OOK, FSK-1 and FSK-2 schemes. The OOK scheme necessitates an equal amount of null and data subcarriers, and the incoming signal is selectively shifted solely to transmit the bit $'1'$. FSK-1 relies on the single-null subcarrier at the edges of the band, while FSK-2 utilizes a larger number of null subcarriers to improve reliability at the expense of reduced spectral efficiency.}
    \label{fig:proposed-schemes}
\end{figure}

\subsection{Proposed Schemes for FSK Modulation}

In this study, we consider two \ac{FSK} strategies discussed in the following subsections.
\subsubsection{FSK-1: Dual-Side Single-Null Subcarrier-Based FSK}
\label{sec:fsk-1}

Let one \ac{BD} bit duration span one \ac{OFDM} symbol duration, then \color{black} the backscatter modulation in \ac{FSK}-1 approach can be expressed as
\begin{equation}
    \backscatterDataSignal [n] = 
    \begin{cases}
        e^{-\frac{j2\pi f_1 n}{\TotalSubcarriers}}, & b=0\\
        e^{\frac{j2\pi f_1 n}{\TotalSubcarriers}}, & b=1
    \end{cases}~.
\end{equation}
 To enable interference-free  \ac{BC} with \ac{FSK}-1 approach, the null subcarriers are allocated between the data subcarriers such that a null subcarrier is placed before the data subcarrier so the \ac{BD} can shift the signal to different frequencies and the \ac{Rx} can detect both the direct link and backscatter signal information at different subcarriers allocated to each of them. According to \ac{FSK}-1, the data symbols are mapped to the subcarriers as
    \begin{equation}
    X[k] = 
    \begin{cases}
        \primaryDataSymbol[m], & k = 2m\\
        0,\; & \mathrm{otherwise}
    \end{cases}~,
\end{equation}
where $m\in\mathbb{Z}^+$ with $m\leq(\TotalSubcarriers/2)-2$. Thus, one subcarrier adjacent to each data subcarrier is left as a null subcarrier.

Upon receiving the \ac{OFDM} signal with this subcarrier allocation, \ac{BD} shifts the signal as follows: If the transmit information bit is $'0'$, \ac{BD} shifts the data subcarriers to the previous null subcarrier. If it transmits $'1'$, \ac{BD} shifts the data subcarriers to the next null subcarriers, as shown in \figurename~\ref{fig:proposed-schemes}. The \ac{FSK}-1 approach is applicable only for the binary \ac{FSK}  modulation because more null subcarriers are needed between the data subcarriers to enable a higher-order \ac{FSK} without direct link interference. For instance, if the \ac{BD} shifts the incident signal with a frequency equal to two subcarriers, the first data subcarrier interferes with the subsequent data subcarrier. 

\subsubsection{FSK-2: Single-Side Multi-Null Subcarrier-Based FSK}
To improve the reliability of {FSK}-based \ac{BC}, we introduce \ac{FSK}-2 that can use multiple null subcarriers, as shown in \figurename~\ref{fig:proposed-schemes}. With \ac{FSK}-2, the data symbols are mapped to the subcarriers as
   \begin{equation}
    X[k] = 
    \begin{cases}
        \primaryDataSymbol[1], & k =1 \; \mathrm{and} \;  
        m=1\\
        \primaryDataSymbol[m], & k = 
         2m+\zeta \; \mathrm{and}\; m>1 \\
        0,\; & \mathrm{otherwise}
    \end{cases}~,
\end{equation}
where $\zeta$ is the number of null subcarriers between two adjacent data subcarriers and $m\leq\lfloor(\TotalSubcarriers-1)/(\zeta+1)\rfloor$ for $m\in\mathbb{Z}^+$. For binary \ac{FSK}, two null subcarriers are allocated between two adjacent data subcarriers (i.e., $\zeta=2$), and the \ac{BD} modulates the incident signal by shifting the data symbols to the subsequent subcarriers. For instance, to transmit the information bit $'0'$, the data symbol is shifted by one subcarrier spacing. On the other hand, the data symbol is shifted by the twice subcarrier spacing to transmit information bit $'1'$, as shown in \figurename~\ref{fig:proposed-schemes}. Moreover, the value of $\zeta$ can be adjusted to achieve a higher order \ac{FSK} modulation and multiple \ac{BD} access. {\color{black}The backscatter modulation in  \ac{FSK}-2 approach can be expressed as}
\begin{equation}
    \backscatterDataSignal [n] = 
    \begin{cases}
        e^{\frac{j2\pi f_1 n}{\TotalSubcarriers}}, & b=0\\
        e^{\frac{j2\pi f_2 n}{\TotalSubcarriers}}, & b=1
    \end{cases}~.
\end{equation}
It is worth noting that \ac{FSK}-1 is spectrally efficient compared to \ac{FSK}-2. However,  it is less reliable if the detector at the RX relies on the single-null subcarrier at the band's edges. \ac{FSK}-2 utilizes a larger number of null subcarriers than \ac{FSK}-1. The inclusion of these additional null subcarriers significantly enhances the reliability of \ac{BC} in the \ac{FSK}-2 scheme, as shown in Section~V. 

\section{Non-Coherent Detection and Performance Analysis}
\label{sec:detection}
The \ac{Rx} receives the signal expressed in  \eqref{Eq:prx_received}. After removing the \ac{CP} and taking the $\TotalSubcarriers$-point \ac{DFT}, the received signal can be written as
\begin{equation}
    \begin{aligned}
    \ReceivedSignalFreq[k] = & \DreceivedDirectSignalFreq[k]+\DreceivedBDSignalFreq[k]+\noiseSignalFreq[k]\\
    = &\directChannelFreq[k]\primaryOFDMSignal[k]+\reflectionCoefficient\forwardChannelFreq[k]\backwardChannelFreq[k]\primaryOFDMSignal[k]\BDdataFreq[k] +\noiseSignalFreq[k]~.
    \end{aligned}
\end{equation}
where $\ReceivedSignalFreq[k]$ denotes the received symbol on the $k$-th subcarrier with $\DreceivedDirectSignalFreq[k]$,  $\DreceivedBDSignalFreq[k]$, $\forwardChannelFreq[k]$, $\backwardChannelFreq[k]$ and $\noiseSignalFreq [k]$ representing   $\DTreceivedDirectSignal[n]$, $\DTreceivedBDSignal[n]$, $\forwardChannel[n]$, $\backwardChannel[n]$ and $\noiseSignalTime[n]$  in the frequency domain, respectively. $\BDdataFreq[k]$ is the Fourier transform of $\backscatterDataSignal[n]$.
Since the data from the \ac{BD} and \ac{BS} are received at the different subcarriers, we can write the received signal from \ac{BS} and \ac{BD} separately as
\begin{equation}
    \begin{aligned}
        \hat{Y}_\mathrm{d} [\hat{k}] &= \ReceivedSignalFreq[\hat{k}]= \directChannelFreq[\hat{k}]\primaryOFDMSignal[\hat{k}]+\noiseSignalFreqDirect[\hat{k}]~,
    \end{aligned}
\end{equation}
and 
\begin{equation}
    \begin{aligned}
        \hat{Y}_\mathrm{b} [\tilde{k}] = &\ReceivedSignalFreq[\tilde{k}]=\reflectionCoefficient\forwardChannelFreq[\tilde{k}]\primaryOFDMSignal[\tilde{k}]\backwardChannelFreq[\tilde{k}]+\noiseSignalFreqBD[\tilde{k}]~,
    \end{aligned}
    \label{Eq:BDrecvSignalSep}
\end{equation}
where $\hat{k} \in \dataSubcarrierSet$ and $\tilde{k}\in \nullSubcarrierSet$ represent the index of subcarriers containing data of the \ac{BS} and \ac{BD}, respectively. According to the the \ac{BD}'s modulation, we can express  $\primaryOFDMSignal[\tilde{k}]$ as 
\begin{equation}
\primaryOFDMSignal[\tilde{k}]=
    \begin{cases}
        0, &b[n]=0\\
        \reflectionCoefficient\primaryOFDMSignal_d[k+\zeta], &b[n]=e^{\frac{-j2\pi\BDFreqShift n}{\TotalSubcarriers}}\\
        \reflectionCoefficient\primaryOFDMSignal_d[k-\zeta], &b[n]=e^{\frac{j2\pi\BDFreqShift n}{\TotalSubcarriers}}
    \end{cases}~.
\end{equation}

\subsection{Non-coherent Detector for OOK}
In the case of binary transmission from the \ac{BD} with \ac{OOK} modulation, $\hat{\ReceivedSignalFreq}_\mathrm{b} [\tilde{k}]$ can be expressed as
\begin{equation}
        \hat{Y}_\mathrm{b} [\tilde{k}] \triangleq 
    \begin{cases}    
    \noiseSignalFreqBD[\tilde{k}],& b = 0\\
    \DreceivedBDSignalFreq [\tilde{k}]+\noiseSignalFreqBD[\tilde{k}],& b = 1
  \end{cases}~,
  \label{Eq:BDreceivedSignal}
\end{equation}
where the $\noiseSignalFreqBD[\tilde{k}]\sim\mathcal{CN}(0,\sigma_{\noiseSignalFreqBD}^2)$. For a large value of $N$, the \ac{OFDM} signal backscatter by \ac{BD} with \ac{OOK} $\DreceivedBDSignalFreq [\tilde{k}]$ is a sequence of random variables and follows the circularly-symmetric complex Gaussian distribution with zero mean and variance $\sigma_\mathrm{Y_b}^2 = |\reflectionCoefficient|^2||H_b|^2\sum_{\tilde{k}=0}^{N_\mathrm{b}-1}|H_f(\tilde{k})|^2$. 
Without the loss of generality, we assume here that $\zeta = 1$. 
$\TotalBDSubcarriers$, and $\TotalDataSubcarriers$ represent the total number of null subcarriers and data subcarriers in an \ac{OFDM} symbol transmitted by the \ac{BS}, respectively. After receiving the signal from \ac{BS} and \ac{BD}, the \ac{Rx} must detect these signals to retrieve the information. A non-coherent detector is used at \ac{Rx} to detect the signal of \ac{BD} while considering that the data of the \ac{Rx} is detected in a conventional manner using a coherent detector because both devices' data are received over the orthogonal subcarriers. Since the non-coherent detector does not require phase information, it can be used in different types of \ac{SR} systems, like a mutualistic \ac{SR} system in which BS and \ac{BD} have a common \ac{Rx} or an ambient \ac{SR} system where they have separate \acp{Rx} (i.e., user equipment and ambient \ac{Rx}). Comparatively, in coherent detection, \ac{Rx} must know  $\directChannelFreq$, $\forwardChannelFreq$, and $\directChannelFreq$, which is suitable for an \ac{SR} with a common \ac{Rx} but not for an \ac{SR} with an ambient \ac{Rx} \cite{liang2020symbiotic}.


To detect the \ac{OOK} signal, \ac{Rx} only needs to know the placement of the data and null-subcarriers in the transmitted \ac{OFDM} signal without $\primaryDataSymbol[\hat{k}]$, $h_\mathrm{d}$, $h_\mathrm{f}$, and $h_\mathrm{b}$. The information in $\BDdataFreq[\tilde{k}]$ transmitted by shifting or not shifting the $\primaryDataSymbol[\hat{k}]$ from $\primaryOFDMSignal[\hat{k}]$ to $\primaryOFDMSignal[\tilde{k}]$ can be acquired by detecting the energy at $\primaryOFDMSignal[\tilde{k}]$ subcarriers.  Let $\mathcal{H}_0$ be the hypothesis for the likelihood of detecting the signal (i.e., $b=1$) at the receiver when no signal was transmitted (i.e., $b=0$). This situation occurs when the received signal energy exceeds a predefined threshold. On the other hand, let $\mathcal{H}_1$ be the hypothesis for the likelihood of detecting no signal, i.e., $b=0$ when the signal was transmitted, i.e., $b=1$. This situation occurs when the received signal energy is below a threshold. To test these two hypotheses, we define a test statistic $r$ based on $\hat{Y}_\mathrm{b}[\tilde{k}]$ 
\begin{equation}
    r = \sum_{\tilde{k}=0}^{\TotalBDSubcarriers-1}\left|\hat{Y}_\mathrm{b}[\tilde{k}]\right|^2~. 
\end{equation}
We can also write the $r$ according to the information transmitted by \ac{BD} as
\begin{equation}
    r = 
    \begin{cases}    
    \sum_{\tilde{k}=0}^{\TotalBDSubcarriers-1}|\noiseSignalFreqBD[\tilde{k}]|^2,& \backscatterDataSignal=0\\
    \sum_{\tilde{k}=0}^{\TotalBDSubcarriers-1}|\forwardChannelFreq[\tilde{k}]\primaryOFDMSignal[\tilde{k}]\backwardChannelFreq[\tilde{k}]+\noiseSignalFreqBD[\tilde{k}]|^2,& \backscatterDataSignal = 1~
  \end{cases}~.
\end{equation}
Under $\mathcal{H}_0$ when \ac{BD} transmits bit $'0'$, the subcarriers allocated for \ac{BC} only contain noise, and $r=r_0$ is the sum of the square of the $2N_b$ Gaussian random variables. We need to find the distribution of $\probability{r|\backscatterDataSignal=0}$ by considering that the square of the sum of two Gaussian random variables (i.e., the real and imaginary part of the noise signal)  follows the exponential distribution. However, to find the distribution of the sum of $N_\mathrm{b}$ exponential random variables, we need to perform $N_\mathrm{b}-1$ convolution operation equivalent to taking $N_\mathrm{b}-1$ integrals, which is a highly complex undertaking. As $W_b$ is a circularly symmetric complex Gaussian random variable, the characteristic function of $|W_b|^2$ is $\left( 1-it \lambda^{-1}\right)^{-1}$, i.e., the Fourier transform of its \ac{PDF}. For the given $r=r_{0}$, we can find the characteristic function as
\begin{equation}
   \characfun{r_{0}}{t} = E\left[e^{itx}\right]=\prod_{\tilde{k}=1}^{N_b}\frac{1}{1-it\lambda_{\tilde{k}}^{-1}}~,
\end{equation}
where $\lambda_{\tilde{k}}^{-1}=2 \sigma_{w,\tilde{k}}^2$. Then, the \ac{PDF} of $\probability{r|B=0}$ is obtained from characteristic function as 
\begin{equation}
    \pdf{r_{0}}{x} = \int_{-\infty}^\infty \frac{1}{2\pi j t}\characfun{r_{0}}{t}e^{-jtz}dt~.
    \label{Eq:pdfr0}
\end{equation}

On the other hand, under $\mathcal{H}_1$, the distribution of $\probability{r|\backscatterDataSignal=1}$ becomes complicated because of the dependence on the cascaded channels $h_f[n]$ and $h_b[n]$. For a given $h_b[n]$, the \ac{PDF} of the cascaded backscatter channel is derived in Lemma \ref{lemma1}.

\begin{lemma}
\label{lemma:1}
    Let $h_b[n]$ have a single path and have a constant value during one \ac{OFDM} symbol due to the short distance between \ac{BD} and \ac{Rx}. Then, considering $v=|h_b|$ the \ac{PDF} of $r=r_1$ becomes 
    \begin{equation}
        \pdf{r_1|v}{x} = \int_{-\infty}^\infty \frac{1}{2\pi j t}\characfun{r_1|v}{t}e^{-jtz}dt~,
        \label{Eq: pdfforward}
    \end{equation}
    where $\characfun{r_1}{t}=\left(1-it\lambda_{\tilde{k}}^{-1}\right)^{-1}$ is the characteristic function, and $\lambda_{\tilde{k}}^{-1}=2\reflectionCoefficient v^2\sum_{\tilde{k}=0}^{\TotalBDSubcarriers-1}\sigma_{h,\tilde{k}}^2+ \sigma_{w}^2$~.
    \label{lemma1}
\end{lemma}
\begin{proof}
    The proof of Lemma \ref{lemma1} is given in Appendix \ref{prof:lemma:pdfOOK}.
\end{proof}

\paragraph{Error Performance}
In the present system design, we measure the performance in terms of the probability of error for the non-coherent detector, which is obtained by summing the \ac{PFA} and \ac{PMD}. The probability of error for given \ac{PFA} and \ac{PMD} is expressed as 
\begin{equation}
    \poe = \frac{1}{2} \pmd+\frac{1}{2}\pfa~.
\end{equation}
We can find the $\pmd$ and $\pfa$ for a given threshold value $\threshold$ as
\begin{equation}
\begin{aligned}
        \pfa(\threshold) &= \probability{r>\threshold|B=0}\\
         &= 1-\probability{r<\threshold|B=0}\\
         & = 1-\cdf{r}{\threshold}~,
\end{aligned}
\label{Eq:pfa}
\end{equation}
and,
\begin{equation}
\begin{aligned}
        \pmd(\threshold) &= \probability{r<\threshold|B=1}.\\
        &= \cdf{r}{\eta}~.
\end{aligned}
\label{Eq:pmd}
\end{equation}
where $\cdf{r}{\eta}$ is the \ac{CDF} of $r$, which can be calculated from \eqref{Eq:pdfr0} according to the inversion formula given in \cite{waller1995obtaining}:
\begin{equation}
    \cdf{r}{\threshold;v} = \frac{1}{2}-\int_{-\infty}^\infty \frac{1}{2\pi j t}\characfun{r_1|v}{t}e^{-jtz}dt~.
     \label{Eq:FrCDFv}
\end{equation}
The decision threshold should be chosen such that it minimizes $\poe$ as
\begin{equation}
    \tilde{\threshold} = \text{arg}\underset{\threshold}{\text{min}} \poe(\eta)~.
    \label{Eq:threshold}
\end{equation}
If we consider a fixed $\pfa$ and try to minimize $\pmd$, we can find  $\tilde{\threshold}$ that minimizes $\poe$. After fixing  $\pfa$, we search for a given $\sigma_{\noiseSignalFreqBD}^2$ by using a one-dimensional linear search to find the optimal threshold  \cite{elmossallamy2019noncoherenta}. Then, we use $\tilde{\threshold}$ to calculate the minimum $\poe$.   

Now, we can find the \ac{CDF} for $v$ as a Rayleigh random variable, which is written as
\begin{equation}
    \begin{aligned}
    \cdf{r}{\eta} &= \int_0^\infty \cdf{r}{\threshold;v} f(v)dv\\
    &= \int_0^\infty \left(\frac{1}{2}-\int_{-\infty}^\infty \frac{1}{2\pi j t}\characfun{r_1|v}{t}e^{-jtz}dt\right) f(v)dv\\
    &=\frac{1}{2}\int_0^\infty f(v)dv-\\&\;\;\;\;\int_0^\infty\int_{-\infty}^\infty \left(\frac{1}{2\pi j t}\characfun{r_1|v}{t}e^{-jtz}dt\right)f(v)dv~,
    \end{aligned}
    \label{Eq:FrCDFbackward}
\end{equation}
where $f(v)=\frac{v}{\sigma_v^2}e^{-\frac{v^2}{\sigma_v^2}}$. $\cdf{r}{\threshold;v}$ is the joint \ac{CDF} of $\threshold$ and $v$, which is obtained from the joint probability distribution of the forward and backward links that follow the  Rayleigh distribution. We aim to compute  $\pmd$ in \eqref{Eq:pmd} for \ac{BC} based on a given $\threshold$. In case the \ac{BD} is close to the \ac{Rx}, the backward link is a single line-of-sight path, $v$ is taken as constant, and \eqref{Eq:FrCDFv} can be used to calculate $\pmd$. Otherwise, we marginalize $\cdf{r}{\threshold;v}$ by integrating over $v$ to find $\cdf{r}{\eta}$. Finally, we calculate the $\poe$ using \eqref{Eq:pmd} and \eqref{Eq:pfa} for a given $\threshold$. As there is no closed-form solution of $\cdf{r}{\eta}$, we can compute the integrals using any software \cite{andrews2011tractable}.

\subsection{Non-Coherent Detector for FSK}
Let $\fskOneset$, and $\fskTwoset$ represent the set of null subcarriers dedicated for \ac{BD} bit $'0'$ and bit $'1'$, respectively. The RX then calculates metrics for information bits $'0'$ and $'1'$ as
\begin{align*}
    \ts{0} &\triangleq \sum_{\tilde{k}\in\fskOneset} |\ReceivedSignalFreq[\tilde{k}]|^2\\
    &= \sum_{\tilde{k}\in\fskOneset}\left|\directChannelFreq[\tilde{k}]\primaryOFDMSignal[\tilde{k}]+\forwardChannelFreq[\tilde{k}]\reflectionCoefficient\primaryOFDMSignal[\tilde{k}]\backwardChannelFreq[\tilde{k}]+\noiseSignalFreq[\tilde{k}]\right|^2~,
\end{align*}
and,
\begin{align*}
        \ts{1} &\triangleq \sum_{\tilde{k}\in\fskTwoset} |\ReceivedSignalFreq[\tilde{k}]|^2\\
        & = \sum_{\tilde{k}\in\fskTwoset}\left|\directChannelFreq[\tilde{k}]\primaryOFDMSignal[\tilde{k}]+\forwardChannelFreq[\tilde{k}]\reflectionCoefficient\primaryOFDMSignal[\tilde{k}]\backwardChannelFreq[\tilde{k}]+\noiseSignalFreq[\tilde{k}]\right|^2~.
\end{align*}
Finally, it detects the \ac{BD}'s information by comparing the values of $\ts{0}$ and $\ts{1}$ as
\begin{equation}
    \tilde{\backscatterDataSignal} = 
    \begin{cases}
        0, & \ts{0}>\ts{1}\\
        1, & \ts{0}<\ts{1}
    \end{cases}~,
    \label{Eq:detected_b}
\end{equation}
where $\tilde{b}$ is the detected \ac{BD} information. In particular, when $\ts{0}$ is greater than $\ts{1}$, the detected information is $'0'$ because the \ac{BD} shifts the primary signal to subcarrier set $\fskOneset$. However, when the \ac{BD} shifts the primary signal to $\fskOneset$, $\ts{0}$ is less than $\ts{1}$ and the \ac{Rx} decides the detected bit as $'1'$. The detection of the transmitted bit at the \ac{Rx} may be incorrect because the received signal is affected by both the channel conditions and the presence of noise. For instance, if \ac{BD} transmits the bit $'0'$, but the value of $\ts{0}$ at the \ac{Rx} is less than $\ts{1}$, the detected bit is decided to be $'1'$  and vice versa. 

The \ac{BER} of the non-coherent detector can be expressed as
\begin{equation}
    P_{\text{e}} = \frac{1}{2}\mathrm{Pr}(\tilde{b}=1|b=0)+\frac{1}{2}\mathrm{Pr}(\tilde{b}=0|b=1)~,
    \label{Eq:poe_fsk}
\end{equation}
where $\mathrm{Pr}(\tilde{b}=1|b=0)$ represents the \ac{PMD} and $\mathrm{Pr}(\tilde{b}=0|b=1)$ denotes \ac{PFA}. We can express these probabilities as
\begin{flalign}
    \mathrm{Pr}(\tilde{b}|b) &\overset{\Delta}{=} 
    \begin{cases}
        \mathrm{Pr}(\ts{0}-\ts{1}\leq0), &b=0\\
        \mathrm{Pr}(\ts{0}-\ts{1}>0), &b=1
    \end{cases}~,\nonumber\\
    & = 
    \begin{cases}
        \cdf{\ts{0}-\ts{1}}{0}, \hfill &b=0\\
        1-\cdf{\ts{0}-\ts{1}}{0}, &b=1
    \end{cases}~,
    \label{Eq:Pr_fsk}
\end{flalign}
Finding the distributions of $\ts{0}$ and $\ts{1}$ become complicated because of the dependence on the cascaded channels $h_f[n]$ and $h_b[n]$. For a given $h_b[n]$, the \ac{PDF} of the cascaded backscatter channel can be found according to Lemma \ref{lemma1}. Moreover,  $\cdf{\ts{0}-\ts{1}}{0}$ denotes \ac{CDF} of $\ts{0}-\ts{1}$, which is obtained as follows:
\begin{theorem}
\rm
The \ac{CDF} of $\ts{0}-\ts{1}$ can be calculated as:
\begin{equation}
    \cdf{\ts{0}-\ts{1}}{0} = \int_0^\infty \cdf{\ts{0}-\ts{1}}{0;v} f(v)\text{d}v~,
     \label{Eq:cdf_fsk}
\end{equation}
where 
\begin{equation}
    \cdf{\ts{0}-\ts{1}}{0;v} = \frac{1}{2} - \int_{-\infty}^{\infty} \frac{\characfun{\ts{0}|v}{t}\characfun{\ts{1}|v}{t}^*}{2\pi jt}  {\text{d}}t~,
    \label{Eq:cdf_fska}
\end{equation}
where $v = |\backwardChannel|$ is the backscatter channel gain with $f(v)=(v/\sigma_v^2)\mathrm{exp}(-v^2/\sigma_v^2)$. The characteristic functions $\characfun{\ts{0}|v}{t}$, and  $\characfun{\ts{1}|v}{t}$ are given by
\begin{equation}
        \characfun{\ts{0}|v}{t}=\frac{1}{1-it\lambda_{\tilde{k}}^{-1}}~,
\end{equation}
and, 
\begin{equation}
    \characfun{\ts{1}|v}{t}=\frac{1}{1-it\lambda_{\tilde{k}}^{-1}}~,
\end{equation}
respectively, for $\lambda_{\tilde{k}}^{-1} = 2\reflectionCoefficient v^2\sum_{\tilde{k}=0}^{\TotalBDSubcarriers-1}\sigma_{h,\tilde{k}}^2+ \sigma_{w}^2$~.
\label{theorem1}
\end{theorem}

\begin{proof}
    The proof of Theorem \ref{theorem1} is given in Appendix \ref{proof:theorem:pdffsk}.
\end{proof}

The calculations of \eqref{Eq:cdf_fsk} and \eqref{Eq:cdf_fska} are valid for the cascaded backscatter channel with forward and backward links following the Rayleigh distributions. We marginalize \eqref{Eq:cdf_fska} over $v$ to find $\cdf{\ts{0}-\ts{1}}{0}$, and compute $\mathrm{Pr}(\tilde{b}|b)$ in \eqref{Eq:poe_fsk}. Although $\cdf{\ts{0}-\ts{1}}{0}$ in \eqref{Eq:cdf_fsk} is not a closed-form expression, it is an analytical solution that can be evaluated via numerical integration, and  $\poe$ can be calculated using \eqref{Eq:poe_fsk} accordingly \cite{sahin2023reliable}.

\begin{figure}
    \centering
    \includegraphics[width=\linewidth]{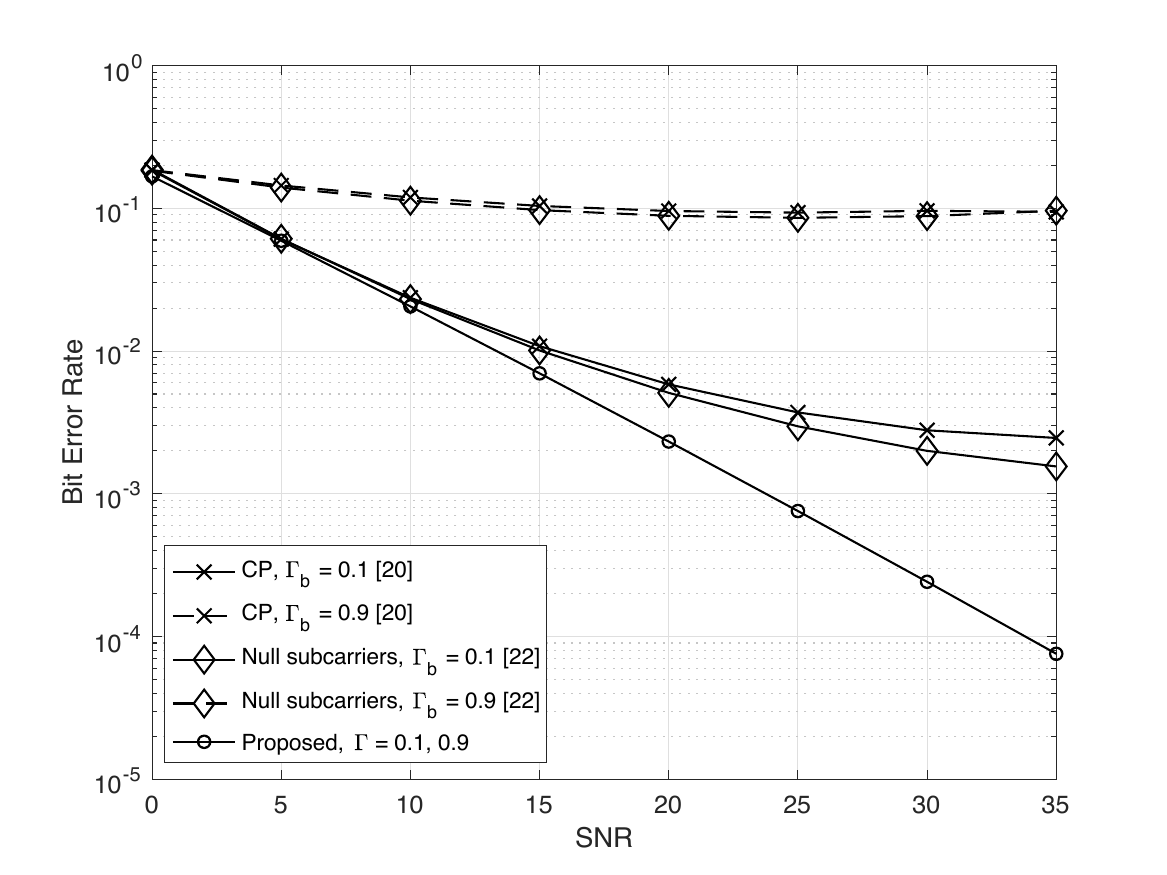}
    \caption{A comparison of the primary communication performance of the proposed scheme with the CP-based and null-subcarrier-based BC schemes in \cite{yang2016backscatter}, and \cite{ elmossallamy2018backscatter}, respectively, when $\reflectionCoefficient \in \{0.1,0.9\}$ and $\TotalSubcarriers = 64$.}
    \label{fig:comparison}
\end{figure}

\subsection{Compatibility and Integration}    
The proposed OFDM-based Symbiotic Radio scheme can be seamlessly integrated with existing communication infrastructure, as the primary signal design is based on established wireless standards such as Wi-Fi, LTE, and 5G, enabling direct compatibility and coexistence \cite{farhang2016ofdm, zhu20193}. 
Additionally, the proposed \ac{OFDM}-based \ac{SR} scheme can be extended to support multiple \acp{BD} using frequency-domain multiplexing. Moreover, the hardware design considerations for the \ac{BD} are discussed in Section II-B, where it is noted that conventional BDs employ rectangular pulse modulation, which can lead to undesirable out-of-band emissions. To mitigate this issue, continuous load modulation controlled by a microcontroller can be used to facilitate the synthesis of complex signals for frequency modulation \cite{kimionis2016pulse}. Furthermore, alternative \ac{BD} architectures have been proposed that can generate single-sideband \ac{FSK} signals without higher-order harmonics \cite{ding2020harmonic, alhassoun2023spectrally}. For the modulation over the primary signal, both the \ac{BD} and the \ac{Rx} must have prior knowledge of the \ac{BS} transmission and be synchronized accordingly. As the Rx is a conventional user in our case, it can get the necessary information about the primary signal and can synchronize with the BS using control signals such as primary synchronization signal (PSS) and secondary synchronization signal (SSS) in LTE and 5G standards. However, the BD may require additional control signaling from the BS and dedicated synchronization mechanisms, such as the example discussed in \cite{feng2023heartbeating}, where the \ac{BD} utilizes the periodicity of LTE signals for synchronization. 

While the proposed approaches facilitate effective interference cancellation and enhance the overall performance of the \ac{OFDM}-based \ac{SR} with a single \ac{BD}, there is significant potential to extend these methods to networks involving multiple \acp{BD}. In such scenarios, the \ac{Rx} must contend with both direct link interference and mutual interference among \acp{BD}. To address these challenges, after appropriate null-subcarrier allocation, each \ac{BD} can shift the primary signal to a distinct subcarrier using the proposed \ac{OOK}, \ac{FSK}-I, and \ac{FSK}-II modulation schemes. 

Although this strategy fosters effective interference cancellation and improves overall network performance, there remains considerable scope for further exploration, for example, under high-mobility environments. The Doppler effect presents challenges for traditional \ac{OFDM} signals, disrupting subcarrier orthogonality and potentially increasing interference. Additionally,  synchronization is a challenging problem in time-varying channels not only for \ac{BC} but also for primary communication. Due to this limitation, our proposed approach may not be directly applicable to high-mobility scenarios when relying on traditional \ac{OFDM} signals. In our future studies, we aim to design a robust receiver architecture that can dynamically adapt to varying conditions in both single \ac{BD} and dense \ac{BD} scenarios. Furthermore, we intend to incorporate adaptive synchronization schemes that take into account the challenges presented by high-mobility environments \cite{chen2009adaptive}.


Lastly, the proposed schemes are presented with a focus on a single-user \ac{OFDM}-based \ac{SR} system to provide a clearer understanding of their benefits. To implement the proposed scheme in a multi-user scenario, we maintain consistent null-subcarrier allocation while assigning subcarriers to primary users based on their specific needs, similar to conventional \ac{OFDM} approaches. The \ac{Rx} only needs to detect the null-subcarrier allocation to identify the \ac{BD}’s signal, and it does not require knowledge of the data on the primary data subcarriers, as our schemes facilitate non-coherent detection. Thus, in a multi-user scenario, the primary subcarrier allocation differs from that in a single-user scenario while ensuring the functionality of our proposed schemes.

\section{Numerical Results}
\label{sec:results}
In this section, we evaluate the proposed scheme numerically. We first study the performance of the primary user in \ac{OFDM}-based \ac{SR} in the presence of \ac{BC} and compare it with the techniques in \cite{yang2016backscatter, elmossallamy2018backscatter}. The baseline scheme in \cite{elmossallamy2018backscatter} utilizes the guard subcarriers for the \ac{BC}, which are left in the \ac{OFDM} symbol in the LTE standard for interference mitigation. Furthermore, the number of guard subcarriers is selected according to the channel bandwidth and subcarrier spacing. For instance, $64$ subcarriers are left in the \ac{LTE} with a $10$ MHz bandwidth \cite{ghosh2010fundamentals}. Since the primary signals in \cite{yang2016backscatter} and \cite{elmossallamy2018backscatter} are not designed symbiotically, they cause interference to data subcarriers, resulting in increasing \ac{BER} for the primary communication. We consider a primary system with \ac{BPSK} modulation for $\TotalSubcarriers\in\{64,128,256,512\}$, and $\CPlength = N/8$. The information at the \ac{BD} is modulated with binary \ac{FSK} modulation. We consider various reflection coefficients with $\reflectionCoefficient \in \{0.1,0.25,0.50,0.75,0.9,1\}$, and the value of $\tilde{\threshold}$ is selected according to \eqref{Eq:threshold} assuming  $\pfa=10^{-3}$. We assume that the direct link, the backward link, and the forward link follow a Rayleigh fading distribution.

In \figurename~\ref{fig:comparison}, we analyze the \ac{BER} performance of different schemes for $\reflectionCoefficient \in \{0.1,0.9\}$ and $\TotalSubcarriers = 64$. The \ac{BER} of both \ac{CP}-based and guard-subcarriers-based approaches in \cite{yang2016backscatter},  \cite{elmossallamy2019noncoherenta}, and \cite{elmossallamy2018backscatter} increases as the value of $\reflectionCoefficient$ increases. For instance, the \ac{BER} of these schemes increases from $10^{-4}$ to $10^{-3}$ at a \ac{SNR} of $30$ dB due to the interference from \ac{BC}. However, for $\reflectionCoefficient\ge0.9$, the \ac{BD} reflects the majority of the received power, and the \ac{BER} for the primary communication in earlier schemes increases drastically. Comparatively, the proposed approaches for \ac{OFDM}-based \ac{SR} prevent the interference from \ac{BC} to primary communication by strategically placing null subcarriers between data subcarriers. \tablename~\ref{tab:01} summarizes the key differences between the proposed approach and the baseline approaches. 

\begin{table}[]
    \centering
        \caption{Comparison with Existing Methods}
    \label{tab:01}
    \begin{tabular}{|p{0.15\linewidth} | p{0.20\linewidth}|p{0.50\linewidth}|}
    \hline
        Methods & Backscatter Modulation & Key Differences \\ \hline
        Proposed & \ac{OOK}, \ac{FSK}-1, and \ac{FSK}-2 &  \ac{BC} over dedicated null-subcarriers in an \ac{OFDM} symbol \\ \hline
         \cite{yang2016backscatter}& BPSK & \ac{BC} over \ac{OFDM} signal with \ac{BD}'s data detection using uncorrupted part of \ac{CP}  \\ \hline 
    \cite{ elmossallamy2019noncoherenta}& OOK&\ac{BC} over the guard subcarriers in the \ac{OFDM} signal\\ \hline
         \cite{elmossallamy2018backscatter}& \ac{FSK} & \ac{BC} over the empty guard subcarriers in the \ac{OFDM} signal\\ \hline
         
    \end{tabular}
\end{table}

\subsection{Simulation Results for OOK Scheme}
In \figurename~\ref{fig:pmd-bd-alpha}, we analyze the \ac{PMD} for non-coherent detection of \ac{OOK} modulated \ac{BC} signals with $\zeta = 1$, $\TotalBDSubcarriers/2$, \ac{PFA} $= 10^{-3}$, $\TotalSubcarriers = 128$, $\TotalBDSubcarriers = 64$ and $\reflectionCoefficient\in\{0.25,0.5,0.75,1\}$. As can be seen from \figurename~\ref{fig:pmd-bd-alpha}, the \ac{PMD} decreases with increasing \ac{SNR} for a specific value of $\reflectionCoefficient$ due to the decrease in the noise power. Furthermore, the increase in the $\reflectionCoefficient$ results in low \ac{PMD}, and the minimum \ac{PMD} is achieved when $\reflectionCoefficient = 1$ because the \ac{BD} reflects the incident signal with the same input power without loss. Specifically, as the \ac{BD}'s reflected power increases from $\reflectionCoefficient = 0.25$ to $\reflectionCoefficient = 1$, the \ac{PMD} decreases from $10^{-2}$ to more than $10^{-3}$ at $30$ dB \ac{SNR}. We compare the performance of the proposed scheme with the baseline \ac{BC} in \cite{elmossallamy2018backscatter}. The proposed scheme achieves the $\pmd$ same as that of the baseline \ac{BC} at $\reflectionCoefficient = 1$ and $25$ dB \ac{SNR} but with a lower reflection coefficient of  $\reflectionCoefficient=0.50$. Also, the theoretical \ac{PMD} results are aligned with the simulation results. The difference is due to the channel correlation. 

\begin{figure}[t!]
    \centering
    \includegraphics[width=\linewidth]{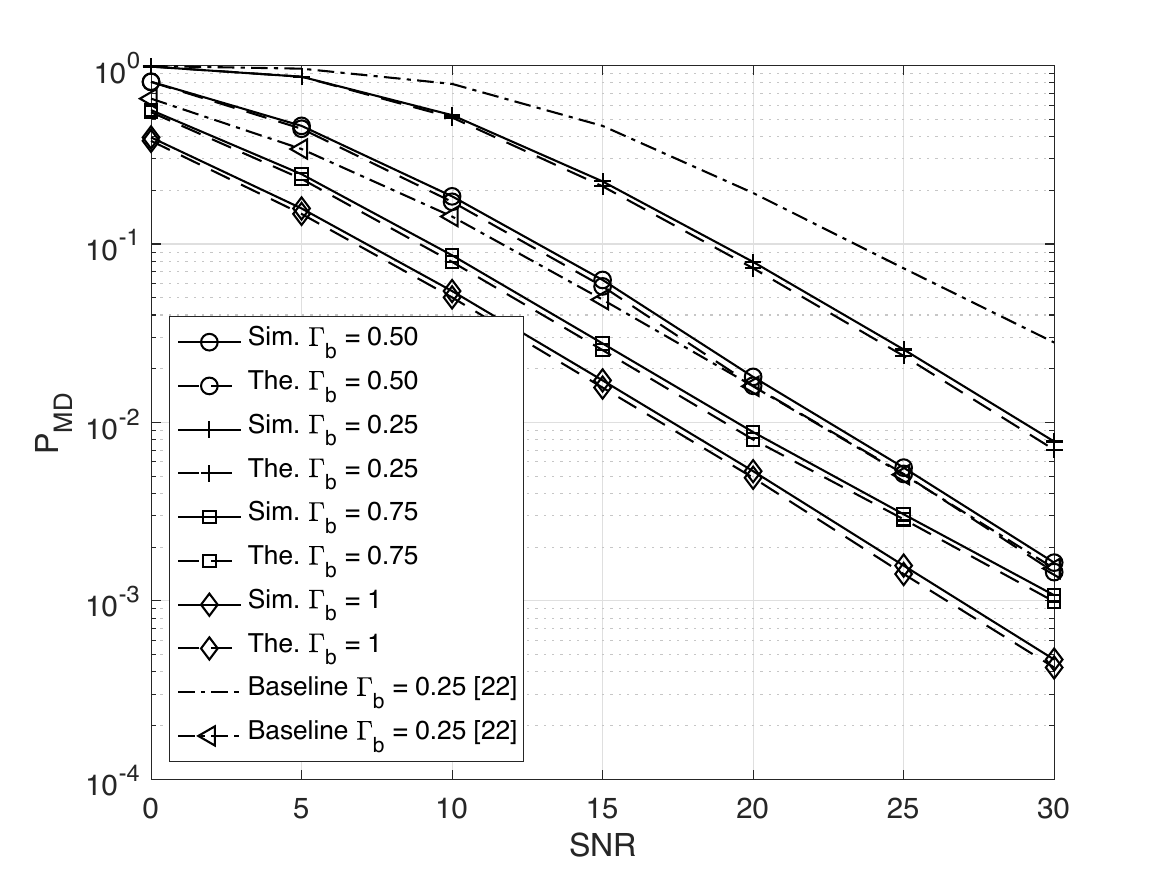}
    \caption{The PMD of BD for different values of SNR obtained from the non-coherent detector with $N=64$ and  $\text{PFA}=10^{-3}$.}
    \label{fig:pmd-bd-alpha}
\end{figure}

In \figurename~\ref{fig:pmd-bd-ng}, we analyze the \ac{PMD} for non-coherent detection of \ac{OOK}-modulated \ac{BC} signals with $\zeta = 1$, $\TotalBDSubcarriers/2$, \ac{PFA} $= 10^{-3}$, $\reflectionCoefficient = 0.25$ and $\TotalSubcarriers \in\{64,128,256,512\}$. From \figurename~\ref{fig:pmd-bd-ng}, it is obvious that the \ac{PMD} depends on the total number of null subcarriers $\TotalBDSubcarriers$ such that as the value of $\TotalBDSubcarriers$ increases from $64$ to $512$, \ac{PMD} reduces to $10^{-2}$ at $20$ dB \ac{SNR}. Compared to the baseline \ac{BC} scheme in  \cite{elmossallamy2018backscatter},   \ac{OOK} achieves significantly lower $\pmd$ for different values of $\TotalSubcarriers$. For instance, the \ac{OOK} scheme achieves $\pmd$ much higher than $10^{-3}$ for  $\TotalSubcarriers = 512$ while the baseline's $\pmd$ is lower than $10^{-3}$ at $30$ dB \ac{SNR}. The simulation results match the theoretical results; however, there is a difference due to channel correlation. Thus, one way to decrease the \ac{PMD} of the \ac{OOK}-modulated \ac{BC} system is to increase $\TotalBDSubcarriers$, which is only possible if the \ac{OFDM} symbol with a large $\TotalSubcarriers$ is transmitted.

\begin{figure}[t!]
    \centering
    \includegraphics[width=\linewidth]{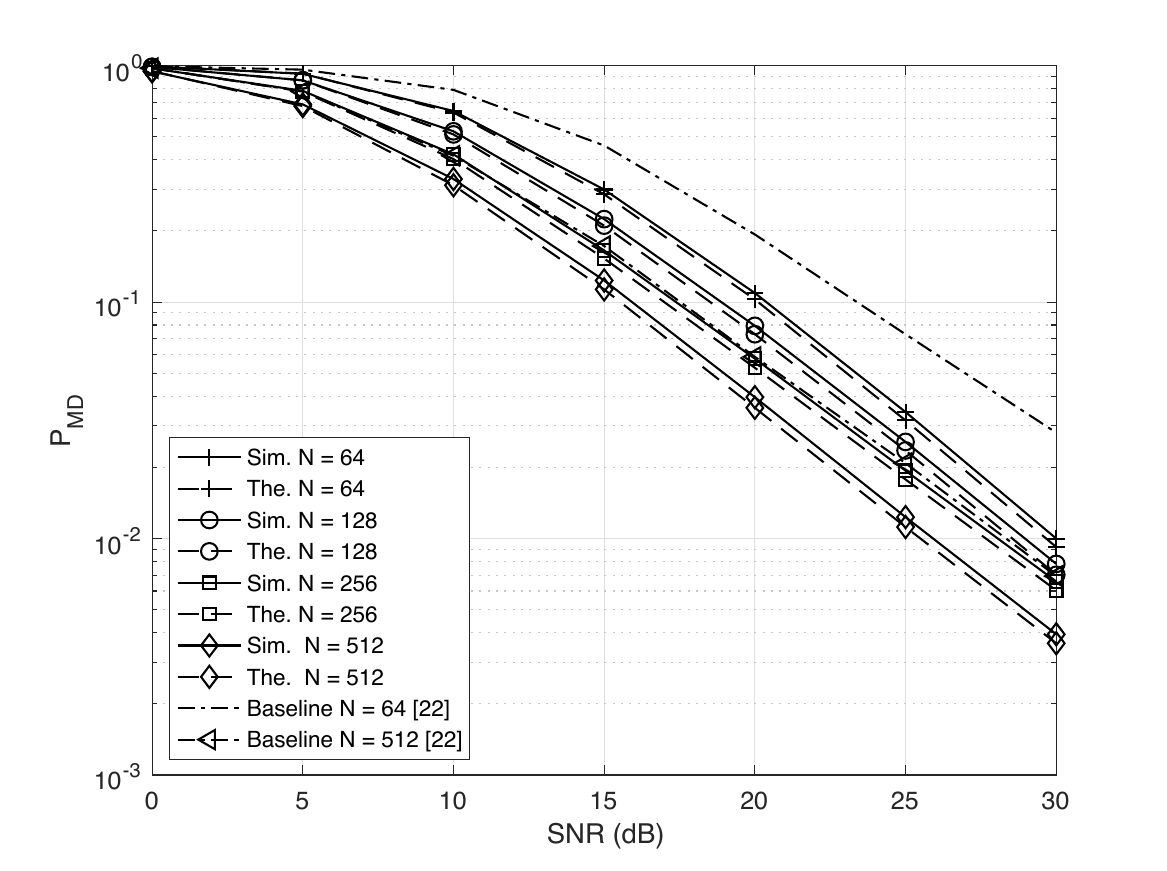}
    \caption{The PMD of BD for different values of SNR and $\TotalSubcarriers \in\{64, 128, 256, 512\}$ obtained from the non-coherent detector with $\reflectionCoefficient = 0.25$.}
    \label{fig:pmd-bd-ng}
\end{figure}

In \figurename~\ref{fig:ROC}, we analyze the \ac{ROC} of the  \ac{OOK} modulated \ac{BC} system with $\reflectionCoefficient = 0.25$, $\TotalSubcarriers=64$, and $\TotalBDSubcarriers=32$ for \ac{SNR} $\in\{0,5,10\}$ dB. The \ac{ROC} curve in \figurename~\ref{fig:ROC} shows that at $0$ dB and $10\%$  \ac{PFA}, around $22\%$ of the \ac{PD} is achieved. As \ac{SNR} increases with $5$ dB, \ac{PD} increases up to $27\%$ with the same value of \ac{PFA}. Furthermore, when the \ac{SNR} is increased by $5$ dB more, the \ac{PD} increases up to $75\%$ at $10\%$ \ac{PFA}. Besides, \ac{ROC} curves obtained through simulations are aligned with the theoretical \ac{ROC} curves for the proposed system with subcarrier allocation for \ac{OOK} modulated \ac{BC}. Thus, the system achieves a higher detection accuracy with the non-coherent detector for \ac{OOK}. 

\begin{figure}[t!]
    \centering
    \includegraphics[width=\linewidth]{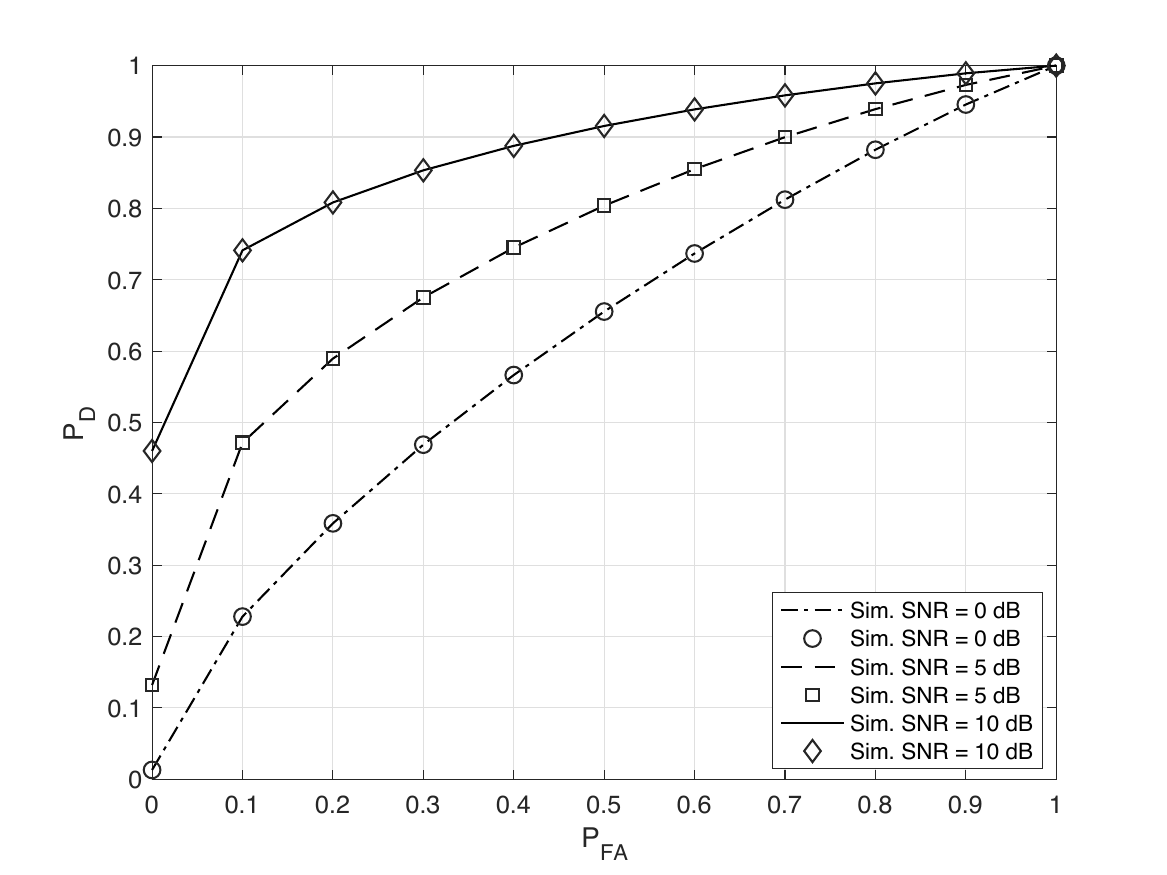}
    \caption{The ROC of the SR system with a non-coherent detector for $\reflectionCoefficient = 0.25$ and $\TotalSubcarriers=64$.}
    \label{fig:ROC}
\end{figure}

\subsection{Simulation Results for FSK Schemes}

In \figurename~\ref{fig:np1_fsk_ber}, we consider \ac{FSK}-1 with $\TotalSubcarriers = 64$, $\TotalBDSubcarriers=\TotalDataSubcarriers+2$, and $\TotalDataSubcarriers=\lfloor(\TotalSubcarriers-1)/2\rfloor$. Out of $\TotalBDSubcarriers= 33$, only two null subcarriers are used to detect the signal according to the subcarrier allocation for \ac{FSK}-1  modulation. The selection of two specific null subcarriers for non-coherent detection is motivated by the fact that the remaining null subcarriers, except for the first and last subcarriers, typically exhibit equal power at the \ac{Rx} for the transmission of both bit $'0'$ and bit $'1'$. Hence, the RX can obtain the transmitted bits based on the power discrepancies observed at the \ac{Rx}. The \ac{FSK}-1 approach enables the non-coherent detector to receive the transmitted data reliably in the presence of the null subcarriers. The system performance is analyzed in terms of \ac{BER} versus \ac{SNR} for different values of $\reflectionCoefficient\in\{0.25,0.50,0.75,1\}$. It can be seen from \figurename~\ref{fig:np1_fsk_ber} that the minimum \ac{BER} is higher than $10^{-2}$ at \ac{SNR} $30$ dB and $\reflectionCoefficient=0.25$ but with increasing the value $\reflectionCoefficient$ to $0.50$, \ac{BER} reduces to $10^{-2}$. However, as the value of the $\reflectionCoefficient$ is increased further to $1$ (i.e., a maximum reflection of power at \ac{BD}), the \ac{BER} goes below $10^{-2}$. The simulation results are also perfectly aligned with the theoretical curves for the proposed non-coherent detection method. We compare the performance of the proposed scheme with the baseline \ac{BC} scheme presented in \cite{elmossallamy2019noncoherenta}, which utilizes the guard subcarriers to transmit the data of the \ac{BD}. The baseline scheme outperforms the \ac{FSK}-1 approach when multiple guard subcarriers are used for \ac{BC}. If only one guard subcarrier is used, FSK-1 achieves performance similar to the baseline scheme. However, the performance gains of \ac{FSK}-1 are primarily derived from the primary communication in the \ac{SR}. In contrast, the primary communication performance degrades significantly in the baseline scheme due to interference between adjacent data subcarriers when \ac{FSK}-based \ac{BC} is employed.

\begin{figure}[t!]
    \centering
    \includegraphics[width=\linewidth]{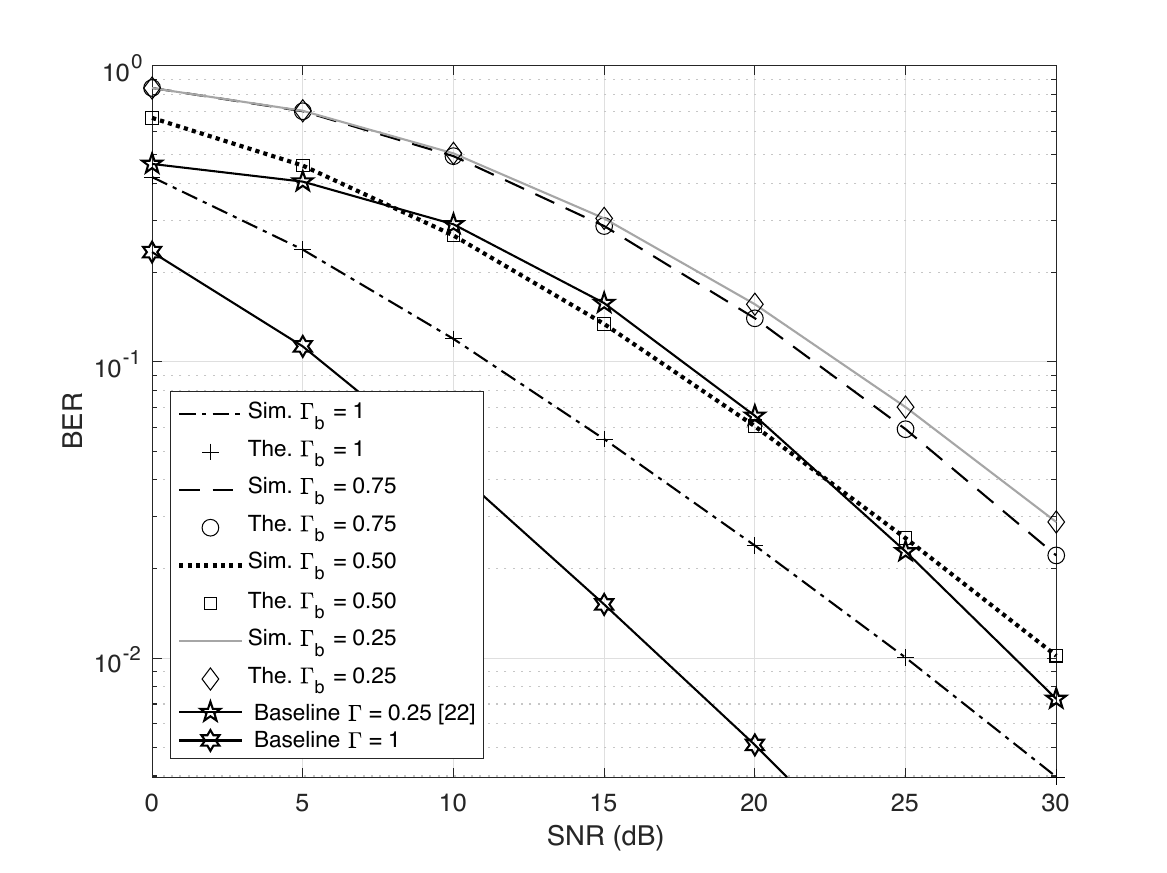}
    \caption{The BER of the non-coherent detector for FSK-1 and $\TotalSubcarriers=64$.}
    \label{fig:np1_fsk_ber}
\end{figure}

\begin{figure}[t!]
    \centering
    \includegraphics[width=\linewidth]{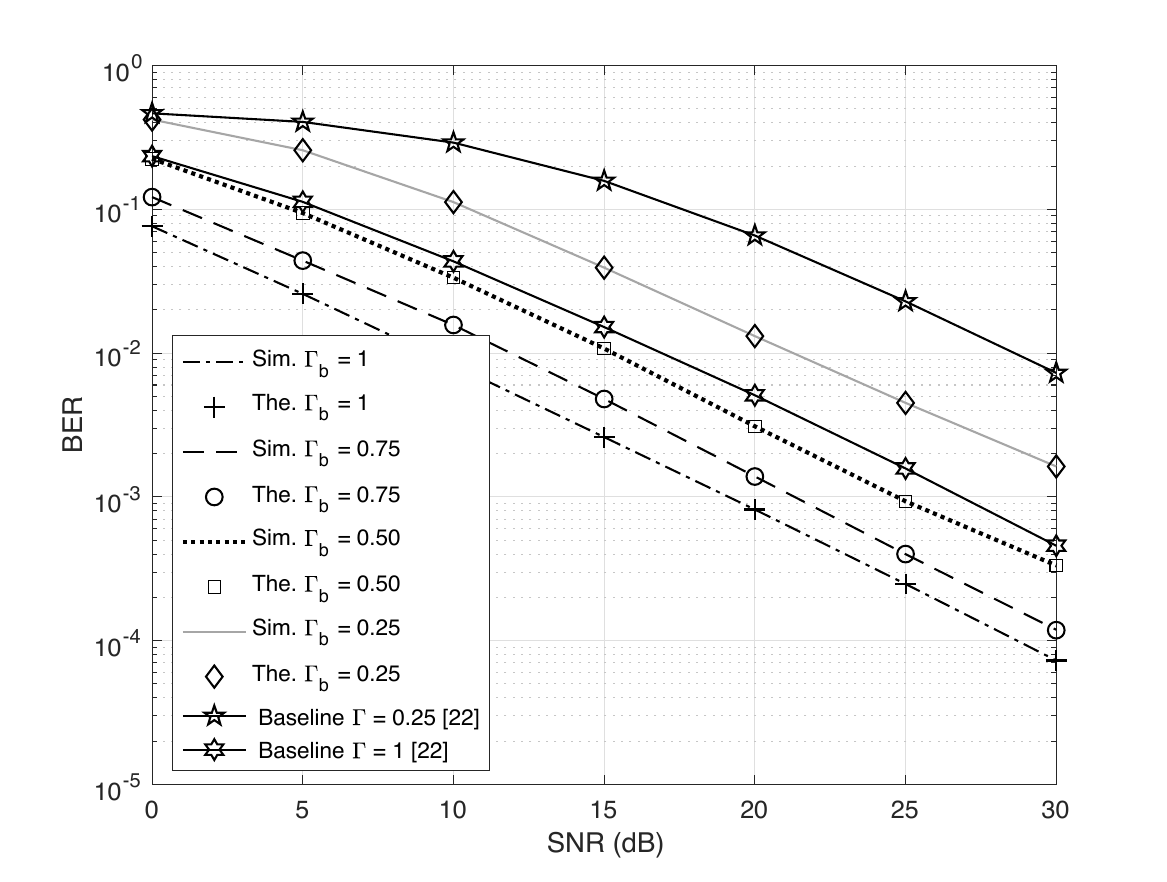}
    \caption{The BER of the non-coherent detector for FSK-2 and $N = 64$.}
    \label{fig:np2_fsk_per_gamma}
\end{figure}

In \figurename~\ref{fig:np2_fsk_per_gamma}, we present the results of non-coherent detection for the proposed  \ac{FSK}-2 scheme in terms of \ac{BER} versus \ac{SNR} with $\TotalSubcarriers = 64$, $\TotalBDSubcarriers=2\TotalDataSubcarriers+1$, and $\TotalDataSubcarriers=(\TotalSubcarriers-1)/3$. We assume that $\TotalBDSubcarriers$ is $42$ and $21$ null subcarriers are used to detect the one information bit transmitted by the \ac{BD}. The \ac{BER} performance is analyzed for different values of $\reflectionCoefficient \in \{0.25,0.50,0.75,1\}$. Selecting the higher value of $\reflectionCoefficient$ results in low \ac{BER} due to an increase in the backscatter signal power. For instance, when \ac{SNR} is $30$~dB, the \ac{BER} decreases notably from $10^{-3}$ to $10^{-4}$ by increasing $\reflectionCoefficient$ from $0.25$ to $1$. Furthermore, the \ac{BER} performance of \ac{FSK}-2 is better than that of \ac{FSK}-1 and the baseline scheme \cite{elmossallamy2019noncoherenta}, as shown in \figurename~\ref{fig:np1_fsk_ber} and \figurename~\ref{fig:np2_fsk_per_gamma}. \ac{FSK}-2 achieves nearly $10^{-1}$ better than the baseline scheme at $\reflectionCoefficient = 0.25$ and $\reflectionCoefficient = 1$ 
 for \ac{SNR} is $30$ dB. The \ac{BER} of \ac{FSK}-2 is $10^{-1}$ times less than \ac{FSK}-1 for $20$~dB \ac{SNR} and $\reflectionCoefficient = 0.25$. The theoretical results also match the simulation results. 

In \figurename~\ref{fig:np2_fsk_per_nfft}, we evaluate the \ac{BER} performance of \ac{FSK}-2 for different values of $\TotalSubcarriers$ and compare it with the baseline approach \cite{elmossallamy2019noncoherenta}. As $\TotalSubcarriers$ increases, $\TotalBDSubcarriers$ also increases, which results in a lower \ac{BER}. For instance, for $20$~dB \ac{SNR}, the \ac{BER} decreases from $10^{-2}$ to $10^{-3}$ when $N$ increases from $64$ to $512$. Therefore, one way to improve the performance of the \ac{SR} with a non-coherent detector for \ac{FSK}-2 \ac{BC} is to increase $\TotalSubcarriers$ without changing  $\reflectionCoefficient$. Furthermore, \ac{FSK}-2 achieves significantly better performance than that of the baseline scheme. For example, at $30$ dB \ac{SNR} value, \ac{FSK}-2 achieves close to $10^{-3}$ where the baseline approach's \ac{BER} is close to $10^{-2}$ for $\reflectionCoefficient=0.25$. 

\begin{figure}[t!]
    \centering
    \includegraphics[width=\linewidth]{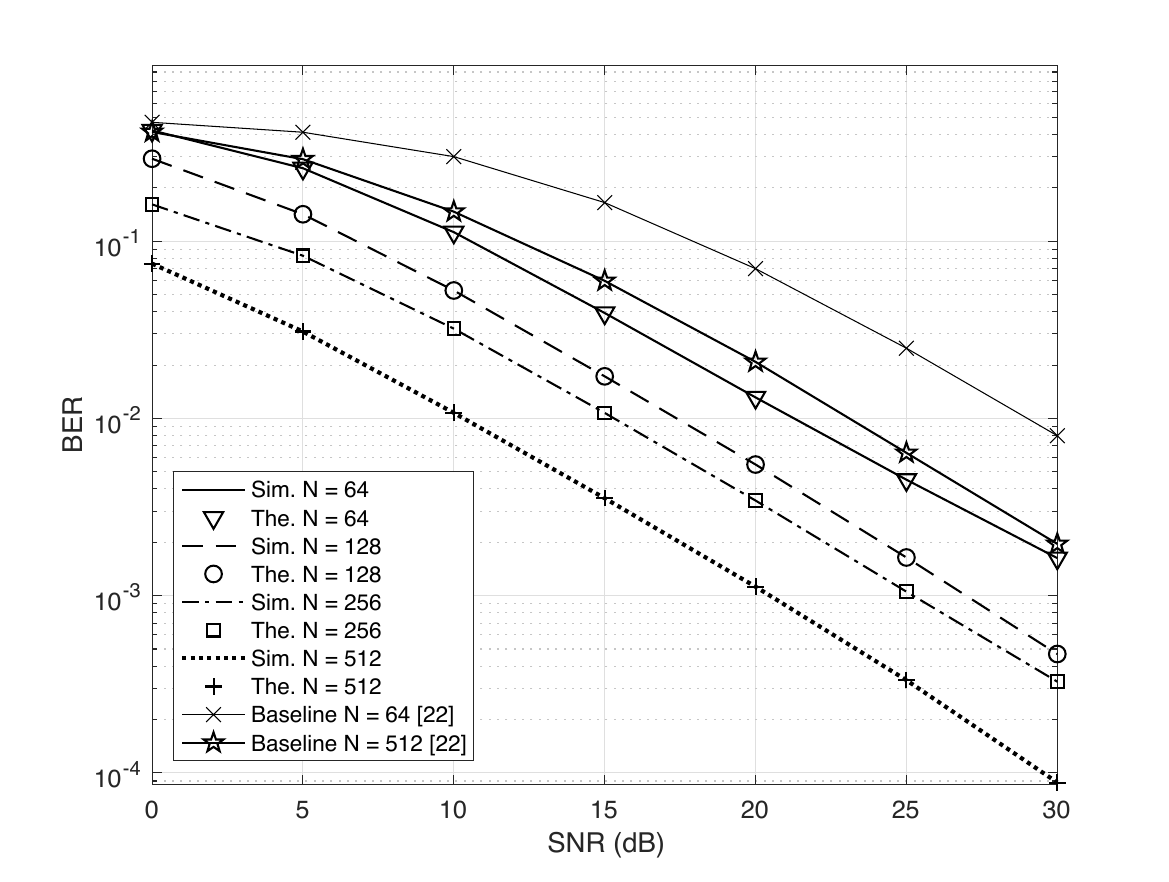}
    \caption{The BER of the non-coherent detector for FSK-2 and $\reflectionCoefficient=0.25$.}
    \label{fig:np2_fsk_per_nfft}
\end{figure}

\begin{figure}[t!]
    \centering
    \includegraphics[width=\linewidth]{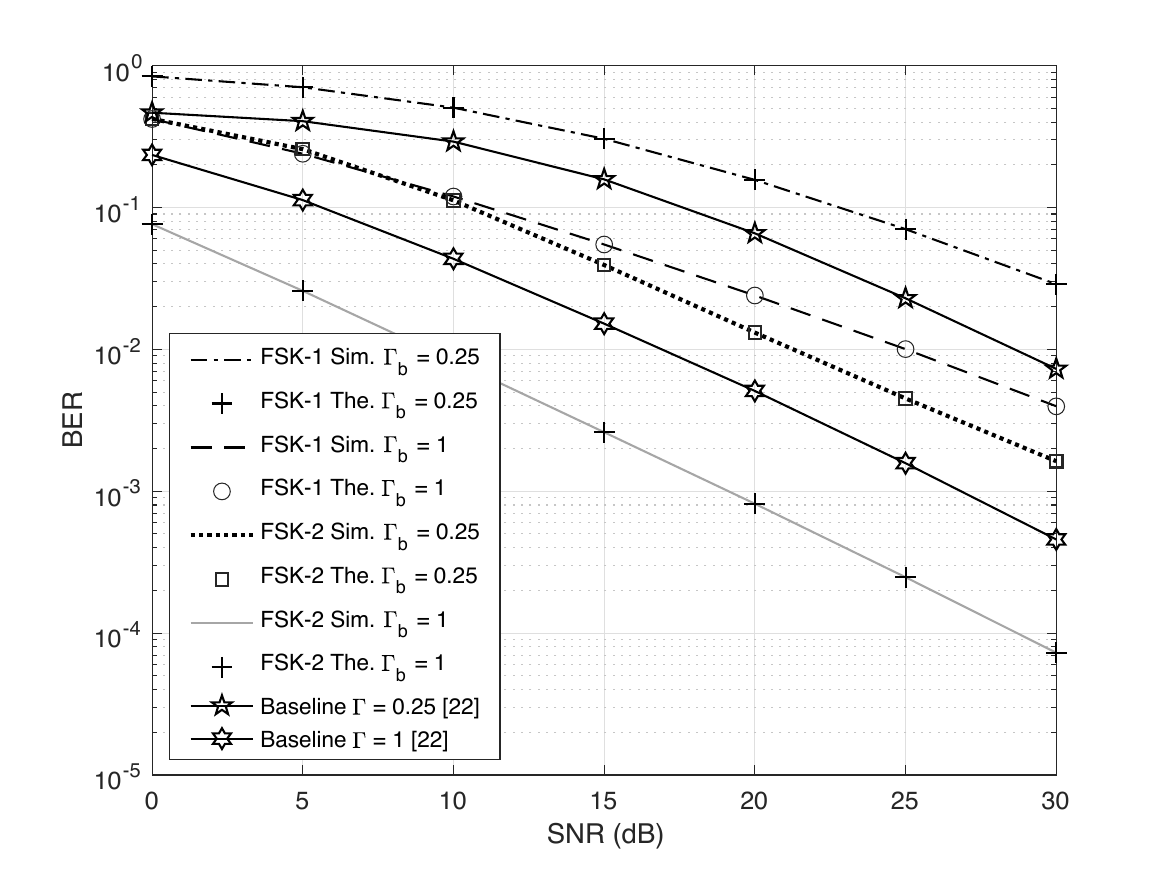}
    \caption{The BER comparison of the non-coherent detector for FSK-1 and FSK-2 with the baseline approach \cite{elmossallamy2018backscatter} when $N = 64$.}
    \label{fig:np12_fsk_per_gamma}
\end{figure}

\begin{figure}
    \centering
    \includegraphics[width=\linewidth]{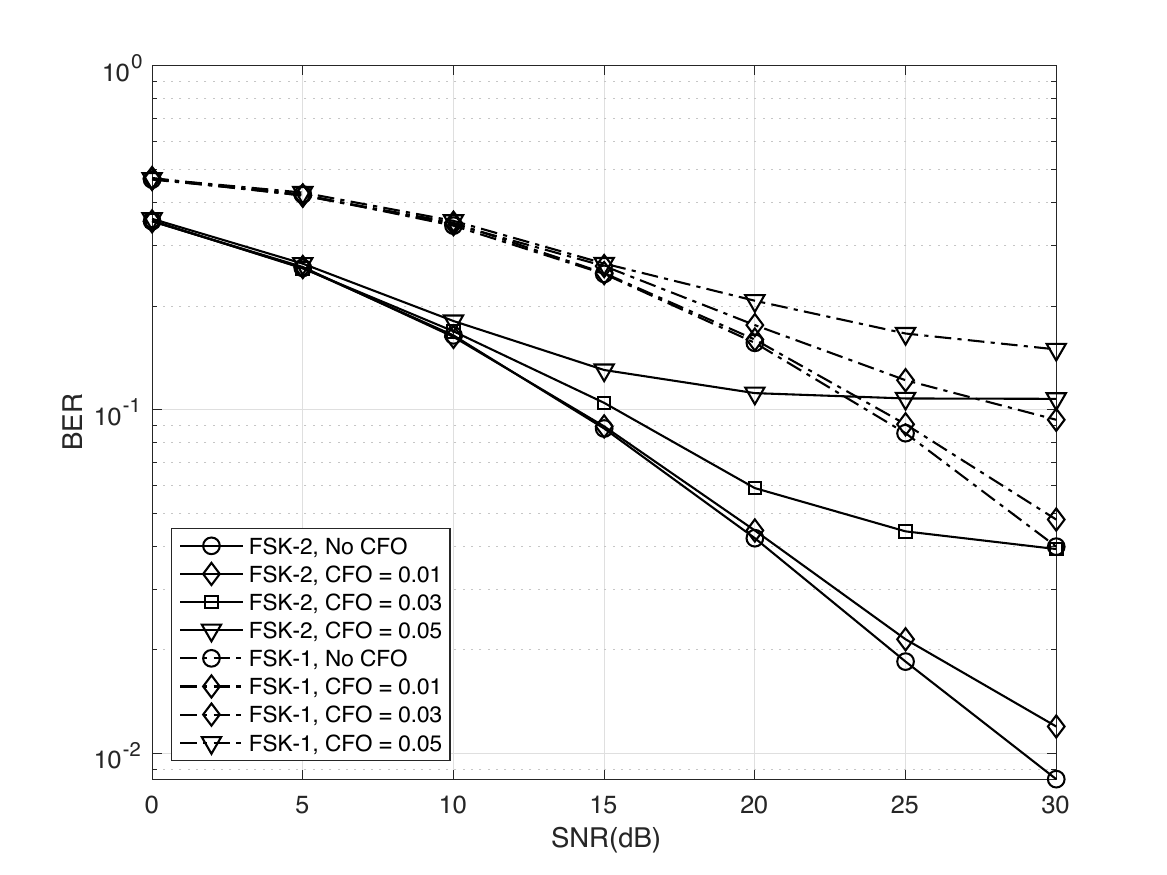}
    \caption{The BER of non-coherent detector for FSK-1 and FSK-2 under CFO impairment.}
    \label{fig:fsk1_fsk2_CFO}
\end{figure}

\begin{figure}
    \centering
    \includegraphics[width=\linewidth]{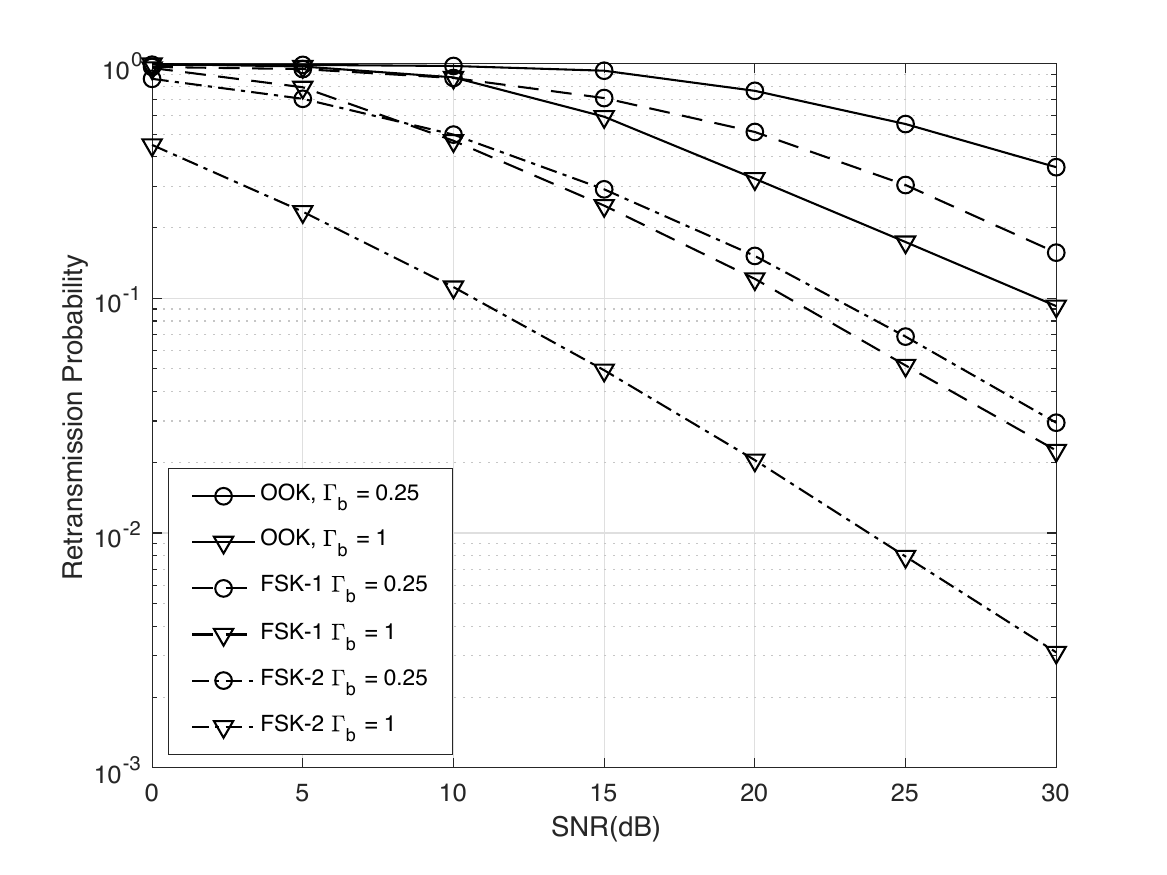}
    \caption{The comparison of the retransmission probabilities of OOK, FSK-1, and FSK-2 with $\TotalSubcarriers = 64$.}
    \label{fig:fsk1sfk2ook_prt_nfft}
\end{figure}

In \figurename~\ref{fig:np12_fsk_per_gamma}, we analyze the performance of \ac{FSK}-1 and \ac{FSK}-2 and compare them with the baseline approach in \cite{elmossallamy2019noncoherenta} for different values of $\reflectionCoefficient$. Similar to \cite{elmossallamy2018backscatter}, the baseline approach uses the guard subcarrier instead of dedicated null subcarriers in the \ac{OFDM} for \ac{BC}. \ac{FSK}-2 outperforms the baseline for different values of $\reflectionCoefficient$ because the value of $\TotalBDSubcarriers$ is greater than the in-band guard subcarriers in an LTE \ac{OFDM}, which can be used for \ac{BC}. For example, the \ac{BER} of \ac{FSK}-2 at $\reflectionCoefficient=0.50$ is higher than the \ac{BER} of the baseline approach at $\reflectionCoefficient = 1$, which means that the \ac{FSK}-2 can achieve better with $50 \%$ reflection of incident power compared to overcome the performance achieved by the baseline with maximum power reflection. Comparatively, \ac{FSK}-1's performance is poor due to the fixed number of null subcarriers for \ac{BC} (i.e., $\TotalBDSubcarriers=1$).    

In \figurename~\ref{fig:fsk1_fsk2_CFO}, we analyze the \ac{BER} performance of the non-coherent detector for different \ac{CFO} values. Since \ac{FSK}-1 and \ac{FSK}-2 schemes shift the incoming signal with a frequency equal to or multiple of the subcarrier spacing (e.g., $15$ kHz in \ac{LTE} and $312.5$ kHz in Wi-Fi). A \ac{CFO} can introduce an extra frequency shift in the primary signal, leading to inter-carrier interference and degrading the quality of the received signal. Hence, the interference between subcarriers causes poor signal detection of primary and \ac{BD} signals at \ac{Rx}. For example, a \ac{CFO} with a value of $0.05$ results in a $10$ times increase in the \ac{BER} of \ac{FSK}-2  at $30$ dB. The performance of \ac{FSK}-1 also degrades due to \ac{CFO} but is comparatively less than that of \ac{FSK}-2. 
The reduced performance of the \ac{FSK}-2 scheme compared to \ac{FSK}-1 can be attributed to the difference in their subcarrier utilization for backscatter signal detection. In the \ac{FSK}-1 scheme, only the edge subcarriers are used for \ac{BD}'s data detection, whereas the \ac{FSK}-2 scheme utilizes all the subcarriers dedicated to \ac{BC}. When \ac{CFO} occurs, the data symbols of the primary communication signal shift and encroach onto the null subcarriers, causing interference. This interference affects the performance of both \ac{FSK}-1 and \ac{FSK}-2 schemes. However, the impact is more pronounced in the \ac{FSK}-2 scheme due to the utilization of all the subcarriers dedicated to \ac{BC}. As the \ac{CFO} increases, the interference from the primary communication data symbols to the null subcarriers becomes more severe. This increased interference has a more detrimental effect on the \ac{FSK}-2 scheme, as all the subcarriers used for \ac{BC} are affected. In contrast, the \ac{FSK}-1 scheme, relying only on the edge subcarriers, is relatively less susceptible to this interference, leading to better performance compared to the \ac{FSK}-2 scheme under large CFO conditions.

While we examine the performance results under \ac{CFO}, it is important to note that interference between data and null subcarriers can also arise due to the Doppler effect in high-mobility environments. The presence of the Doppler effect may impair the performance of the non-coherent detector, leading to a decrease in the overall reliability of the system. Specifically, rapid changes in distance from the transmitter fluctuate the effective \ac{SNR} due to varying path loss and multipath effects, which can further heighten the likelihood of error, consequently elevating the \ac{PMD} and \ac{BER} for \ac{OOK} and \ac{FSK}, respectively.   

\subsection{Retransmission Probability}
In this section, we evaluate the retransmission probability (i.e., the ratio of the retransmitted frames to the total number of transmitted frames) for the non-coherent detector for the proposed schemes based on \ac{CRC}. We assume a $5$-bit \ac{CRC} encoder/decoder according to the \ac{RFID} Gen2 standard \cite{epcglobalradio}. The \ac{CRC} bits are generated using the polynomial $x^5+x^3+1$. \ac{BD} append \ac{CRC} bit to the information bits before transmission. The  \ac{Rx} receives the backscatter signal and applies the  \ac{CRC} decoder after \ac{CP} removal, \ac{IDFT}, demodulation, and non-coherent detection processes. If the \ac{CRC} fails, the \ac{Rx} requests the retransmissions of the data frame through the \ac{ARQ} protocol.

In \figurename~\ref{fig:fsk1sfk2ook_prt_nfft}, we present the retransmission probability versus \ac{SNR} for the proposed \ac{OOK}, \ac{FSK}-1, and \ac{FSK}-2 schemes with non-coherent detection. The frame length is $12$ \ac{BD} bits including $5$ \ac{CRC} bits. The performance is evaluated for  $\TotalSubcarriers = 64$. It can be seen from \figurename~\ref{fig:fsk1sfk2ook_prt_nfft} that the \ac{FSK}-1 provides a lower retransmission probability than \ac{OOK} at \ac{SNR} greater than $10$~dB. However, \ac{FSK}-2 outperforms both \ac{OOK} and \ac{FSK}-1 by achieving the lowest retransmission probability at different values of $\reflectionCoefficient\in\{0.25, 1\}$. Therefore, to achieve a smaller number of retransmissions in \ac{SR} for \ac{BC}, the \ac{FSK}-2 scheme is preferred even at low $\reflectionCoefficient$ with fixed $\TotalSubcarriers$.  

\subsection{Computational Complexity}
\label{sec: computational complexity}
To analyze the computational complexity of the proposed schemes, we need to consider both the primary signal design and information modulation at the \ac{BD}. At the \ac{BS}, the primary signal design involves allocating subcarriers to the \ac{Rx} and \ac{BD} within an OFDM symbol. This aspect of the design has a computational complexity comparable to conventional \ac{OFDM} systems used in LTE and 5G standards. The computational complexity of the conventional OFDM and our proposed transmitter's complexity comes from \ac{FFT} operations, which is  $\mathcal{O}(\TotalSubcarriers \log \TotalSubcarriers)$ \cite{farhang2016ofdm}.

When evaluating the complexity associated with different modulation schemes at the BD, it is important to consider how each scheme affects the operations required for information transmission. For conventional OOK, the BD only needs to shift the signal when transmitting a bit '1'; thus, the complexity is minimal. In contrast, for conventional FSK, the BD must perform frequency shifts for both bits '1' and '0', which increases complexity. The computational complexity of the BD, therefore, depends on the proposed modulation scheme used. For the proposed FSK-1 and FSK-2 schemes, which involve $\Omega$ levels of impedance switching, the complexity is $\mathcal{O}(\Omega)$. For the proposed OOK scheme, the impedance switching complexity is $\mathcal{O}(1)$ when transmitting a bit '0' and $\mathcal{O}(\Omega)$ when transmitting a bit '1'. However, since the maximum complexity required for either bit is $\mathcal{O}(\Omega)$, the overall complexity for OOK also aligns with $\mathcal{O}(\Omega)$. This means that while the proposed OOK scheme may offer simplicity in certain cases, its complexity remains comparable to that of FSK due to the impedance switching requirements.

\subsection{Implementation Complexity and Energy Efficiency}
 As discussed in \ref{sec: computational complexity}, the proposed \ac{OOK}, \ac{FSK}-I, and \ac{FSK}-II modulation schemes require frequency shifts in the primary signal at the \ac{BD} to transmit information bits. The \ac{OOK} modulation involves turning the BD on and performing frequency shift only to transmit bit '1' while  FSK-I and FSK-II need frequency shift for both bits '1' and '0'. Implementing frequency shift at the \ac{BD} entails $\Omega$ levels of impedance switching, which can be implemented using the available hardware designs \cite{iyer2016inter, wang2020low}. Nothing comes for free, higher levels of impedance switching consume more power at the \ac{BD}. Therefore, the added complexity of impedance switching in the \ac{FSK} schemes makes them less energy efficient compared to \ac{OOK}. However, efficient integrated circuits can be employed to enhance energy efficiency across all modulation methods.

\section{Concluding Remarks}
\label{sec:conclusion}

This study presents a method for interference-free \ac{BC} in an \ac{OFDM}-based \ac{SR} system. We design the primary \ac{OFDM} signal to eliminate direct link interference by allocating subcarriers for the \ac{BD} and \ac{Rx} within the \ac{OFDM} symbol, with the \ac{BD} using \ac{FSK} modulation over dedicated null subcarriers. We explore three subcarrier allocation strategies: \ac{OOK}, \ac{FSK}-1, and \ac{FSK}-2. While \ac{FSK}-1 is more spectrally efficient, it has a higher \ac{BER}. In contrast, \ac{FSK}-2 and \ac{OOK} offer greater reliability at the cost of spectral efficiency, with \ac{FSK}-2 exhibiting a lower retransmission probability.
To receive \ac{BD} data without interference, we designed a non-coherent detector for the \ac{Rx}. Our theoretical and simulation results indicate that increasing the \ac{OFDM} symbol size or reflection coefficient significantly reduces the \ac{BER} for \ac{BC}. The proposed method is also beneficial for non-cooperative \ac{SR} systems where the primary user and \ac{BC}
receiver are separated, allowing for independent signal detection using coherent and non-coherent
detectors. This advancement has the potential to improve \ac{IoT} deployments. Additionally, the proposed
solution can reduce interference in various applications in various fields such as retail, healthcare, logistics,
and supply chain management, where \ac{IoT} devices rely on data transmission via \ac{OFDM}-based systems like
Wi-Fi, LTE, and 5G.  However, the proposed schemes face challenges related to spectral efficiency due to dedicated resource allocation, which limits subcarriers for primary communication and creates a trade-off between \ac{BC} reliability and data rate. Future research could explore simultaneous multiplexing of multiple \acp{BD} and users to manage mutual and direct-link interference in \ac{OFDM} systems, leveraging interference detection methods such as energy detection, matched filtering, and spectrum sensing. Additionally, future studies might develop energy harvesting analysis considering varying conditions and fluctuations from the Wi-Fi, LTE, or 5G signals. Lastly, enhancing the security of \ac{OFDM}-based \ac{SR} systems is essential to protect both \ac{BC} and primary communication, potentially using lightweight cryptographic techniques alongside physical layer security. The 

\appendices

\section{The Proof of Lemma~\ref{lemma:1}}
\label{prof:lemma:pdfOOK}
\begin{proof}
The received signal is the function of the product of two random variables: $|h_b|$, a Rayleigh random variable, and the other channel $|h_f|^2$, an exponential random variable. Let, 
\begin{equation}
    u = \sum_{m=1}^{N_b}|h_f[m]|^2,
\end{equation}
be the sum of the exponential random variables while the characteristic function of an exponential random variable  $f(x;\lambda_m)=\lambda_m e^\lambda_m x$ for  $\lambda_m$ is
\begin{equation}
    \Phi_m(t)=\lambda_m\int e^{(it-\lambda_m)z}dx=\left(1+it\lambda^{-1}\right)^{-1},
\end{equation}
where $\Phi_m(t)$ is the Fourier transform of $f(x;\lambda_m)$. 

To find the distribution of $u$, we need to find the convolution of $N_b$ exponential random variables, which is equivalent to the product of its characteristic functions 
\begin{equation}
    \begin{split}
         \Phi_u &= \int e^{itz}(\lambda_1e^{\lambda_1z}*\lambda_1e^{\lambda_2z}*\cdots*\lambda_{N_b}e^{\lambda_{N_b}z}) dz\\
         &= \prod_{m=1}^{N_b}\frac{1}{1+it\lambda_m}.
    \end{split} 
\end{equation}
Now, the \ac{PDF} of $u$ can be found by taking the inverse Fourier transform of $\Phi_u(t)$ as
\begin{equation}
    f_{u}(u) = \int_{-\infty}^\infty \frac{1}{2\pi j t}\prod_{m=1}^{N_b}\frac{1}{1+it\lambda^{-1}} dt.
\end{equation}
\end{proof}

\section{The Proof of Theorem }
\label{proof:theorem:pdffsk}
\begin{proof}
The received signal is the function of the product of two random variables: $ |h_b|$, a Rayleigh random variable, and the other channel $|h_f|^2$, an exponential random variable. Let,  
be expressed as
\begin{equation}
    u = \sum_{m=1}^{N_b}|h_f[m]|^2~, 
\end{equation} 
be the sum of the exponential random variables while the characteristic function of an exponential random variable  $f(x;\lambda_m)=\lambda_m e^\lambda_m x$ for  $\lambda_m$ is
\begin{equation}
    \Phi_m(t)=\lambda_m\int e^{(it-\lambda_m)z}dx=\left(1+it\lambda^{-1}\right)^{-1}~,
\end{equation}
where $\Phi_m(t)$ is the Fourier transform of $f(x;\lambda_m)$. 

To find the distribution of $u$, we need to find the convolution of $N_b$ exponential random variables, which is equivalent to the product of its characteristic functions 
\begin{equation}
    \begin{split}
         \Phi_u &= \int e^{itz}(\lambda_1e^{\lambda_1z}\times\lambda_1e^{\lambda_2z}\times\cdots\times\lambda_{N_b}e^{\lambda_{N_b}z}) \text{d}z\\
         &= \prod_{m=1}^{N_b}\frac{1}{1+it\lambda_m}~.
    \end{split} 
\end{equation}
Now, the \ac{PDF} of $u$ can be found by taking the inverse Fourier transform of $\Phi_u(t)$ as \cite{sahin2023reliable}
\begin{equation}
    f_{u}(u) = \int_{-\infty}^\infty \frac{1}{2\pi j t}\prod_{m=1}^{N_b}\frac{1}{1+it\lambda^{-1}} \text{d}t~.
\end{equation}
Hence, we obtain the \ac{CDF} of $\ts{0}-\ts{1}$ as
\begin{equation}
    \begin{aligned}
    &\cdf{\ts{0}-\ts{1}}{\eta} = \int_0^\infty \cdf{\ts{0}-\ts{1}}{\threshold;v} f(v)\text{d}v\\
    &= \int_0^\infty \left(\frac{1}{2} - \int_{-\infty}^{\infty} \frac{\characfun{\ts{0}|v}{t}\characfun{\ts{1}|v}{t}^*}{2\pi jt}  {\text{d}}te^{-jtz}\right) f(v)\text{d}v\\
    &=\frac{1}{2}\int_0^\infty f(v)\text{d}v-\int_0^\infty\int_{-\infty}^\infty\frac{\characfun{\ts{0}|v}{t}\characfun{\ts{1}|v}{t}^*}{2\pi jt}e^{-jtz}\text{d}tf(v)\text{d}v~.
    \end{aligned}\nonumber
\end{equation}
\end{proof}



\end{document}